\newif\ifTwoColumn
\newif\ifTechReport

\TechReporttrue    

\ifTwoColumn
\documentclass[journal]{IEEEtran}
\else
\ifTechReport
\documentclass[journal,12pt,onecolumn,draftcls]{IEEEtran}
\else
\documentclass[journal,11pt,onecolumn,draftclsnofoot]{IEEEtran}
\fi
\fi

\usepackage[utf8]{inputenc} 
\usepackage[T1]{fontenc}    
\usepackage{url}            
\usepackage{booktabs}       
\usepackage{amsfonts,amsmath,amsthm,amssymb,mathrsfs}       
\usepackage{bbm}
\usepackage{cite}
\usepackage{mathtools,xparse}
\usepackage[makeroom]{cancel}
\usepackage[acronym]{glossaries}
\usepackage{glossaries-prefix}
\usepackage{color}
\usepackage{balance}
\usepackage{enumerate,enumitem}
\usepackage{algpseudocode}
\usepackage{subcaption}
\usepackage{nicefrac}       
\usepackage{microtype}      
\usepackage{xcolor}         
\usepackage{graphicx}
\usepackage[prependcaption,colorinlistoftodos]{todonotes}
\usepackage{multirow}
\usepackage{hyperref}
\usepackage{booktabs,tabularx}
\usepackage{cases}
\usepackage[ruled,vlined]{algorithm2e}
\usepackage{xifthen}
\usepackage{eurosym}
\usepackage{chngcntr}

\graphicspath{{./figures/}}

\DeclareMathSizes{10}{10}{6}{4}

\newcommand{\norm}[1]{\left\lVert#1\right\rVert}

\newcommand{\continuanceref}{}

\newcommand{\R}{\mathbb{R}}

\newcommand{\bbS}{\mathbb{S}}
\newcommand{\mc}[1]{\mathcal{#1}}

\newcommand{\eqdef}{\coloneqq}
\newcommand{\reqdef}{\eqqcolon}

\newcommand{\prob}[2][]{
	\ifthenelse{
		\isempty{#1}}{\mathbb{#2}}{\mathbb{#2}\left\{#1\right\}
	}
}
\newcommand{\expected}[2]{\mathbb{E}_{#1}\left[#2\right]}
\newcommand{\wexpected}[2]{\widehat{\mathbb{E}}_{#1}\left[#2\right]}
\newcommand{\wexpectedj}[2]{\widehat{\mathbb{E}}^j_{#1}\left[#2\right]}

\newcommand{\diag}{\mathrm{diag}}

\newcommand{\col}{\mathrm{col}}

\newtheorem{theorem}{Theorem}[section]
\newtheorem{definition}[theorem]{Definition}
\newtheorem{proposition}[theorem]{Proposition}
\newtheorem{lemma}[theorem]{Lemma}
\newtheorem{corollary}[theorem]{Corollary}
\newtheorem{example}[theorem]{Example}

\newtheorem{remark}[theorem]{Remark}
\newtheorem{standing}[theorem]{Standing Assumption}

\newcommand{\normaltext}[1]{\textnormal{#1}}

\usepackage[colorinlistoftodos]{todonotes}

\usepackage{xcolor,calc}

\makeglossaries
\newacronym{iid}{i.i.d.}{independent and identically distributed}
\newacronym{wrt}{w.r.t.}{with respect to}
\newacronym{ISS}{ISS}{input-to-state stability}
\newacronym{dISS}{$\delta$-ISS}{incremental input-to-state stability}
\newacronym{iIOS}{i-IOS}{incremental input-output stability}
\newacronym{IO}{IO}{input-output}
\newacronym{LTI}{LTI}{linear time-invariant}
\newacronym{SA}{SA}{stochastic approximation}
\newacronym{SAA}{SAA}{sample-average approximation}
\newacronym{IQC}{IQC}{incremental quadratic constraint}
\newacronym{SDP}{SDP}{semidefinite programming problem}
\newacronym{LQR}{LQR}{linear-quadratic regulator}
\newacronym{GD}{GD}{gradient descent}
\newacronym{NAG}{NAG}{Nesterov accelerated gradient}
\newacronym{HB}{HB}{heavy ball}
\newacronym{KRR}{KRR}{kernel ridge regression}
\newacronym{RKHS}{RKHS}{reproducing kernel Hilbert space}
\newacronym{ERM}{ERM}{empirical risk minimization}
\newacronym{LS}{LS}{least-squares}

\newglossaryentry{LMI}
{
	name={LMI},
	description={linear matrix inequality},
	first={\glsentrydesc{LMI} (\glsentrytext{LMI})},
	plural={LMIs},
	descriptionplural={linear matrix inequalities},
	firstplural={\glsentrydescplural{LMI} (\glsentryplural{LMI})}
}
\newglossaryentry{MI}
{
	name={MI},
	description={matrix inequality},
	first={\glsentrydesc{MI} (\glsentrytext{MI})},
	plural={MIs},
	descriptionplural={matrix inequalities},
	firstplural={\glsentrydescplural{MI} (\glsentryplural{MI})}
}

\begin{document}
	
\title{Generalization as a robust performance property of learning-enabled dynamical systems}
\author{Filippo Fabiani
\thanks{F. Fabiani is with the IMT School for Advanced Studies Lucca, Piazza San Francesco 19, 55100, Lucca, Italy ({\tt filippo.fabiani@imtlucca.it}).}}


\maketitle

\begin{abstract}
By focusing on algorithmic stability as a means of establishing  out-of-sample bounds, we provide a system-theoretic interpretation of generalization in learning-enabled dynamical systems arising in data-driven optimization and feedback control approximation. Given two neighboring datasets, we specifically model sample replacement as an exogenous disturbance acting on a sensitivity system, while the incremental behavior of the data-dependent operator is encoded through an integral quadratic constraint. By relying on dissipativity arguments, we establish a matrix inequality-based certificate and a uniform stability bound that separates the one-sample sensitivity of the learned operator, and an algorithm-dependent dynamical gain. The latter can then be optimized, offering a tractable tool for certifying and comparing generalization capabilities of learning dynamics. We show that our results recover classical ones for gradient descent, apply naturally to momentum-based methods such as heavy-ball and Nesterov acceleration, and extend to data-driven control. 
\end{abstract}

\IEEEpeerreviewmaketitle

\glsresetall

\section{Introduction}\label{sec:intro}
\IEEEPARstart{I}{n} machine learning, \emph{generalization} refers to the ability of a certain learning-based model to perform well on new, unseen data, rather than only on the samples used for training \cite{shalev2014understanding}. Thus,
ensuring excellent generalization capabilities is recommended in modern dynamical systems whose evolution is governed by a learning policy $\hat\Phi(\cdot, \mc S_N)$ trained on dataset $\mc S_N$ with $N$ samples. A prominent example includes iterative methods for data-driven optimization and \gls{ERM} \cite{vapnik1991principles,kleywegt2002sample}, which may overfit in case of poor generalization. Another suitable instance stems from feedback control approximation, where data-driven proxies are employed to steer the system evolution \cite{care2023kernel,fabiani2023reliably}. Here, poor out-of-sample capabilities may lead to, e.g., degraded performance, or even instability.

In this context, \emph{algorithmic stability} \cite{bousquet2002stability} is a central tool to explain the generalization behavior of learning algorithms. In its classical form, stability quantifies the change in an algorithm's output, measured through a prescribed loss function, when a single training sample is replaced. When this change is small uniformly over neighboring datasets $\mc S_N$ and $\mc S_N^j$, where the latter differs from the former in the $j$-th sample, one obtains high-probability generalization guarantees through concentration arguments \cite{mcdiarmid1989method}. 
A standard approach to perform the stability analysis of iterative procedures consists in deriving algorithm-specific recursive bounds on the distance between two trajectories generated from neighboring datasets, and then rely on the contraction/nonexpansive property of the data-driven mapping $\hat\Phi(\cdot, \mc S_N)$\cite{hardt2016train,chen2018stability,charles2018stability,mou2018generalization,farnia2021train,ho2025instability,fabiani2026finite,fabiani2026certifying}. 
Although effective, such arguments do not explicitly capture how the data perturbation is attenuated or amplified by the algorithmic dynamics over iterations, thereby obscuring the inherent \gls{IO} structure underlying the algorithmic stability mechanism.

This is precisely the perspective we offer. By focusing on a broad class of discrete-time learning dynamics driven by a data-dependent operator $\hat\Phi(\cdot, \mc S_N)$, we compare the system trajectories corresponding to two neighboring datasets $\mc S_N$ and $\mc S_N^j$, obtaining a \emph{data sensitivity system} whose induced operator discrepancy
can be decomposed into two components. The first consists of the sensitivity propagated through $\hat\Phi(\cdot, \mc S_N)$ itself, while the second one is the extra sensitivity injected due to sample replacement. Acting as an exogenous disturbance channel, the latter enables our interpretation of algorithmic stability as a robust performance property from this data perturbation channel to the loss function.
Our analysis hence relies on dissipativity arguments, where the available information on the regularity of $\hat\Phi(\cdot, \mc S_N)$ is encoded through an \gls{IQC} \cite{yakubovich1971procedure}, while the effect of replacing one sample is captured by a uniform data-injection bound. A storage function then certifies a finite gain from the data perturbation to the terminal (i.e., after $T$ steps) state sensitivity. Capitalizing on this, we develop an optimization-based procedure to quantify the algorithmic stability bound, and hence generalization capabilities, of the learning-enabled system at hand. In particular, the (uniform) stability coefficient decomposes in two main terms:
i) the one-sample sensitivity of the data-dependent operator, which scales as $1/N$ in several relevant settings, and ii) an algorithm-dependent dynamical sensitivity gain certified by a dissipativity inequality. Such a decomposition is our main conceptual message: generalization depends on both how strongly a single sample perturbs $\hat\Phi(\cdot, \mc S_N)$, and how the closed-loop dynamics dissipates or amplifies the injected sensitivity.

\subsection{Related work}
To the best of our knowledge, only a few recent works have considered algorithmic stability under a system-theoretic lens. 

Of particular relevance, \cite{attia2021algorithmic} focused on the stability of \gls{NAG} method, 
and established a sharp distinction between quadratic and general smooth convex costs. 
As a main conclusion, it turned out 
that acceleration may strongly amplify small perturbations and that convergence guarantees alone do not determine the algorithmic stability of an optimization dynamics. The present work, instead, offers a complementary vision by formulating the propagation of data perturbations as an \gls{IO} property of a general sensitivity system. In our framework, the amplification phenomena identified in \cite{attia2021algorithmic} are naturally reflected in the dynamical gain from the data injection channel to the terminal hypothesis.

In \cite{kozachkov2023generalization} it was shown, instead, that a contractive iterative scheme in a suitable Riemannian metric is uniformly stable, obtaining bounds proportional to the inverse of both the contraction rate and the number of samples. Although similar in the spirit, our approach differs in the robustness property being certified. In fact, \cite{kozachkov2023generalization} derived stability from the contraction of neighboring trajectories, while in this paper sample replacement is explicitly represented as an exogenous input, and uniform stability is obtained from a dissipativity certificate bounding the resulting induced effect. This allows one to separate the sensitivity of $\hat\Phi(\cdot, \mc S_N)$ to the dataset from the ability of the dynamics to dissipate or amplify the resulting perturbation.

In \cite{li2026unified} the authors studied uniform stability of first-order accelerated schemes in the smooth, strongly convex quadratic regime. Here, accelerated methods were represented as Lur'e-type feedback interconnections between a linear dynamical system and a gradient oracle, encoded smoothness and strong convexity through sector \glspl{IQC}, and searched for quadratic Lyapunov functions via \glspl{LMI}. 
Besides adopting a similar state-space and \gls{IQC} machinery, our contribution starts from a different decomposition of the data sensitivity dynamics, including the disturbance channel caused by the sample replacement, which is not absorbed into an optimization-specific recursion. This also leads to a broader scope. Unlike \cite{li2026unified}, our analysis is indeed not limited to first-order optimization, since $\hat\Phi(\cdot, \mc S_N)$ may correspond to the empirical gradient or a data-driven monotone operator evaluated along the algorithmic trajectory or, more generally, an operator learned from data before the dynamics is deployed, for instance through \gls{LS} or \gls{KRR}.

\subsection{Summary of contributions and paper organization}
Our framework allows one to establish a modular connection between statistical generalization and robust control, as it essentially shows that a learning-based dynamics generalizes when the sensitivity injected by a one-sample replacement is sufficiently dissipated by the closed-loop interconnection. 

In summary, this paper makes the following contribution:
\begin{itemize}
    \item[i)] 
    We recast algorithmic stability of learning-enabled dynamics
	as an \gls{IO} robustness property. Specifically, we show that uniform stability can be interpreted as a bounded-gain property from the exogenous disturbance generated by sample replacement to the loss discrepancy;

    \item[ii)]
    We derive a dissipativity condition, in the form of a \gls{MI}, that allows one to certify such a bounded-gain property.
	This yields a uniform stability bound that
	separates the statistical sensitivity of $\hat\Phi(\cdot, \mc S_N)$ from the amplification or dissipation induced by the dynamics;

    \item[iii)]
    We formulate the computation of the dynamical sensitivity gain as a \gls{SDP} complemented by a scalar search, thereby providing a systematic way to optimize and compare the stability property, and hence generalization, of different dynamics;

	\item[iv)]
	We show that our framework recovers standard contractivity-based stability mechanisms for \gls{GD}, while also applying naturally to lifted methods with memory, such as \gls{HB} and \gls{NAG}, for which contractivity may fail. Interestingly, numerical experiments show that nominal convergence and data sensitivity are distinct design objectives.
\end{itemize}
In addition, we show that several mappings $\hat\Phi(\cdot, \mc S_N)$ of practical interest are characterized by one-sample perturbation scaling as $1/N$, allowing one to obtain consistent, exponential, high-probability bounds. Finally, we emphasize how our results apply to a broad class of data-dependent operators $\hat\Phi(\cdot, \mc S_N)$, including those learned offline via \gls{LS} \cite{fabiani2026concentration} or \gls{KRR} \cite{care2023kernel}, and illustrate them on a kernel-based residual-controller example.

The rest of the paper is organized as follows: in \S \ref{sec:problem_formulation} we formalize the problem addressed, provide basic notions of algorithmic stability and a system-theoretic connection. In \S \ref{sec:sensitivity_approach} we develop a dissipativity-based approach to optimize uniform stability, which can be interpreted as a robust control problem, and discuss the computational aspects. In \S \ref{sec:data_based_algorithms} and \ref{sec:learning_based_dyn}, instead, we apply our technique to characterize the generalization properties of popular data-based optimization and learning-based dynamics, also performing numerical experiments. Finally, in Appendix~\ref{app:1overN} are reported examples of common operators characterized by one-sample perturbation scaling as $1/N$.


\section{Problem formulation and preliminaries}\label{sec:problem_formulation}
Throughout this paper, we will consider the following control-affine, nonlinear, discrete-time dynamics:
\begin{equation}\label{eq:dynamics}
	\left\{
	\begin{aligned}
		x_{k+1}&=Ax_k+B\hat\Phi(y_k,\mc S_N),\\
		y_k&=C x_k + D \hat\Phi(y_k,\mc S_N),
	\end{aligned}
	\right.
\end{equation}
where the vector of state variables $x\in\R^n$ evolves according to an autonomous term through $A\in\R^{n\times n}$, and a deterministic mapping $\hat\Phi:\R^p\times\Xi^N\to\R^m$, premultiplied by $B\in\R^{n\times m}$. The output equation determining the measurement vector $y\in\R^p$, instead, is obtained with matrices $C\in\R^{p\times n}$ and $D\in\R^{p\times m}$.

Given a training set $\mc S_N\eqdef\{\xi^{(i)}\}_{i=1}^N$ consisting of $N$ \gls{iid} samples $\xi^{(i)}\in\Xi\subseteq\R^d$ with \emph{unknown} statistics $\prob{P}$, the operator $\hat\Phi$ either directly manipulates these data as in, e.g., \gls{SAA} and kernel-based techniques, or amounts to a data-driven map learned offline via, e.g., \gls{LS} or \gls{KRR}. 

For the well-posedness of the implicit interconnection occurring in the output of \eqref{eq:dynamics}, we will assume the following:

\begin{standing}\label{standing:system_dynamics}
	For every $N\ge1$, $\mc S_N\in\Xi^N$ and $x\in\R^n$, $y=C x+D \hat{\Phi}\left(y, S_N\right)$ admits a unique solution.
	\hfill$\square$
\end{standing}

While trivially satisfied when $D=0$, also verifying that $y\mapsto Cx+D\hat\Phi(y,\mathcal S_N)$
is a contraction uniformly in $x$ allows one to meet the requirement in Standing Assumption~\ref{standing:system_dynamics}.

Note that the intrinsic dependency of $\hat\Phi$ on $\mc S_N$ may produce substantially different evolutions for the data-enabled dynamics in \eqref{eq:dynamics}. We will thus explore the \emph{generalization} properties possessed by the underlying dynamics, and establish rigorous certificates on the sensitivity of the resulting trajectories \gls{wrt} changes in the training set $\mc S_N$.
Specifically, we will analyze algorithmic stability as a possible means for generalization, and develop suitable connections bridging this concept with standard system-theoretic dissipativity notions. 

\subsection{Preliminaries on algorithmic stability}
We define a \emph{learning algorithm} as a mapping $\mc A_N :\R^n\times\Xi^N \to \R^n$, taking some $N$-multisample and returning a hypothesis $H_N \eqdef \mc A_N(\mc S_N) \in \R^n$. 
By identifying $\mc A_N(\mc S_N)$ with the dynamics in \eqref{eq:dynamics} observed for $T$ iterations starting from some $x_0$, we simply have $H_N=x_T$.
We call $\mc S_N^{j}$ \emph{neighboring} set, as it is obtained by \emph{replacing} the $j$-th element of $\mc S_N$, i.e., $\mc S_N^{j} = \{\xi^{(1)}, \ldots, \xi^{(j-1)}, \xi', \xi^{(j+1)}, \ldots, \xi^{(N)}\}$, where $\xi'\in\Xi$ is drawn according to $\prob{P}$, \gls{iid} \gls{wrt} $\mc S_N\setminus\{\xi^{(j)}\}$. In this case, $H_{N^{j}} \eqdef \mc A_{N}(\mc S_N^{j})$, which in respect to \eqref{eq:dynamics} yields $H_{N^{j}}=x_T'$. 

The accuracy of the algorithm prediction (i.e., the hypothesis produced) is generally measured according to some \emph{loss function} $\ell:\R^n\times\Xi\to\R_{\ge 0}$.
The \emph{generalization error} is typically attached to the loss, and coincides with the following random variable (the hypothesis $x_T$ indeed depends on $\mc S_N$):
\begin{equation}\label{eq:risk}
	r(\mc A_N)=\expected{\prob{P}}{\ell(H_N, \xi)}.
\end{equation}
Note that computing \eqref{eq:risk} would require $\prob{P}$, which is unavailable in the considered framework. Nevertheless, the simplest estimator for \eqref{eq:risk} amounts to the \emph{empirical error}, defined as:
\begin{equation}\label{eq:empirical_risk}
	\hat r(\mc A_N, \mc S_N)=\frac{1}{N} \sum_{j=1}^{N} \ell(H_N, \xi^{(j)}).
\end{equation}

\begin{definition}[\textup{\hspace{-.02cm}\cite[Def.~6]{bousquet2002stability} Uniform stability}]\label{def:unif_stab}
	An algorithm $\mc A_N$ has \emph{uniform stability} $\eta\ge0$ \gls{wrt} the loss function $\ell$ if for all $N\ge1$, $\mc S_N\in\Xi^N$, and $j\in\{1,\ldots,N\}$, 
	\[
		\underset{\xi\in\Xi}{\textnormal{\textrm{sup}}}~|\ell(H_N,\xi)-\ell(H_{N^j},\xi)| \le \eta.
	\]
	\hfill$\square$
\end{definition}

In words, an algorithm is said to be uniformly stable if changing one sample $\xi^{(j)}$ in the training dataset does not substantially alter the output \gls{wrt} $\ell$. In general, the smaller the $\eta$ the better the stability features possessed by the algorithm at hand.
The following result links the uniform stability and the generalization properties of $\mc A_N$, by providing a probabilistic bound on the \emph{generalization gap}, $r(\mc A_N)-\hat r(\mc A_N,\mc S_N)$:
\begin{lemma}[\textup{\hspace{-.02cm}\cite[Th.~3.2]{kutin2002almost}}]\label{lemma:exp_bound}
	Let $\mc A_N$ be $\eta$-uniformly stable \gls{wrt} a loss function $0\le\ell(H_N,\xi)\le M$, $M\ge0$, for all $\xi\in\Xi$ and $\mc S_N\in\Xi^N$. Then, for any $N\ge1$ and $\delta\in(0,1)$, it holds that:
	\ifTwoColumn
\else
		\[
		\prob{P}^{N}\left\{\mc S_N\in\Xi^N:r(\mc A_N)-\hat r(\mc A_N,\mc S_N)\le\eta+(N\eta+M)\sqrt{2\ln(1/\delta)/N} \right\}\ge1-\delta.
		\]
	\fi
	\hfill$\square$
\end{lemma}
Specifically, Lemma~\ref{lemma:exp_bound} provides with an upper bound on the risk associated to hypothesis $H_N=\mc A_N(\mc S_N)$, which holds with arbitrary confidence $1-\delta$, such that i) if $\eta$ scales proportional to $1/N$, vanishes as $N\to\infty$, and ii) it is exponential in the confidence parameter $\delta$. The former condition is  called \emph{consistency} property, and provides a ``tight'' certificate in the sense that the upper bound in Lemma~\ref{lemma:exp_bound} vanishes as $N\to\infty$.

\subsection{An incremental \gls{IO} interpretation of algorithmic stability}
Algorithmic stability is thus an attractive property, as it allows one to directly verify the finite difference condition required to apply the McDiarmid's concentration inequality \cite{mcdiarmid1989method} and, consequently, to derive generalization bounds. However, the application of this notion to the data-enabled dynamics in \eqref{eq:dynamics}, and that obtained with $\mc S_N^j$ in place of $\mc S_N$, turns into:
\[
	\underset{\xi\in\Xi}{\textrm{sup}}~|\ell(x_T,\xi)-\ell(x_T',\xi)| \le \eta,	
\]
which basically concentrates on a ``snapshot'' of the last step of the resulting $T$-long trajectories, passed through $\ell$. From a system-theoretic perspective, this has limited interpretation, since the available information on the transient, which depends on how sensitive is the learned mapping $\hat\Phi$ to data is wasted. 

With this regard, we note that \gls{iIOS} \cite{sontag1999notions
} is the system-theoretic stability notion closer in the spirit to Definition~\ref{def:unif_stab}. 
In a generic nonlinear framework, one then considers the following \gls{IO} dynamics:
\begin{equation}\label{eq:IOdynamics}
	\left\{
		\begin{aligned}
			x_{k+1}&=\Psi(x_k,u_k),\\
			y_k&=h(x_k),
		\end{aligned}
	\right.
\end{equation}
with $\Psi:\R^n\times\R^m\to\R^n$ and $h:\R^n\to\R^p$. 

\begin{definition}
	The system in \eqref{eq:IOdynamics} is \gls{iIOS} if there exist class-$\mathcal{KL}$ function $\theta$ and class-$\mathcal{K}_{\infty}$ function $\gamma$\footnote{See \cite{kellett2014compendium} for standard definitions of class-$\mathcal{KL}$ and $\mathcal{K}_{\infty}$ functions.} such that:
	\[
		\norm{y_k-y'_k}\le\theta(\norm{x_0-x'_0},k)+\gamma\left(\underset{0\le i\le k}{\mathrm{sup}}~\norm{u_i-u'_i}\right),
	\]
	at every $k\ge0$, where the output trajectories are generated by \eqref{eq:IOdynamics} for any $x_0$, $x'_0$ and input sequences $u$, $u'$. 
	\hfill$\square$
\end{definition}

\begin{proposition}\label{prop:iIOS_uniformstab}
	For any $y\in\R^p$, $N\ge1$, $\mc S_N\in\Xi^N$ and $j\in\{1,\ldots,N\}$, let $\|\hat\Phi(y,\mc S^j_N)-\hat\Phi(y,\mc S_N)\|$ be bounded. If $D=0$ and \eqref{eq:dynamics} is \gls{iIOS} \gls{wrt} the loss-output map $z=h(x)\eqdef\ell(x,\cdot)$, then it is also uniformly stable.
	\hfill$\square$
\end{proposition}
\begin{proof}
	The control-affine nonlinear dynamics in \eqref{eq:dynamics} with $D=0$ can be equivalently absorbed by $\hat\Psi(x_k,\mc S_N)+Bu_k\eqdef Ax_k+B(\hat\Phi(Cx_k,\mc S_N)+u_k)$ with $u_k=0$ for all $k\ge0$, which yields \eqref{eq:IOdynamics} once redefined the output equation as $z_k\eqdef h(x_k)$, with $h:\R^n\to\mc Y$ denoting the loss-output mapping, formally defined as $h(x)\eqdef\ell(x,\cdot)$. Specifically, $\mc Y$ is the space of bounded real-valued functions on $\Xi$, equipped with supremum norm, namely $\norm{h(x)}_{\mc Y}\eqdef\textrm{sup}_{\xi\in\Xi}~|\ell(x,\xi)|$.
	
	Then, pick $j\in\{1,\ldots,N\}$, and consider the following dynamics compatible with $\hat\Psi(x_k,\mc S_N)+Bu_k$ and $\mc S_N^j$:
\[
	\left\{
		\begin{aligned}
			x'_{k+1}&=\hat\Psi(x'_k,\mc S^j_N)=\hat\Psi(x'_k,\mc S_N)+Bu'_k,\\
			z'_k&=h(x'_k),
		\end{aligned}
	\right.
\]
initialized at some $x'_0=x_0$, with $u'_k\eqdef\hat\Phi(Cx'_k,\mc S^j_N)-\hat\Phi(Cx'_k,\mc S_N)$. Notice that, by definition, $\norm{h(x)-h(x')}_{\mc Y}=\textrm{sup}_{\xi\in\Xi}~|\ell(x,\xi)-\ell(x',\xi)|$.  Then, if \eqref{eq:dynamics} is \gls{iIOS} with the redefined output equation, one can readily conclude that it is also uniformly stable, since at the $T$-th iteration, $H_N=\mc A_N(\mc S_N)=x_T$, thereby obtaining:
\[
\begin{aligned}
	\norm{z_T-z'_T}_{\mc Y}&=\norm{h(x_T)-h(x'_T)}_{\mc Y}=\underset{\xi\in\Xi}{\textrm{sup}}~|\ell(x_T,\xi)-\ell(x'_T,\xi)|\\
	&\le\gamma\left(\underset{0\le k\le T}{\textrm{sup}}~\|B(\hat\Phi(Cx'_k,\mc S^j_N)-\hat\Phi(Cx'_k,\mc S_N))\|\right),
\end{aligned}
\]
which, in view of the properties of the function of class $\mc K_\infty$, yields a finite bound on the loss discrepancy in case $\|\Phi(Cx'_k,\mc S^j_N)-\hat\Phi(Cx'_k,\mc S_N)\|$ is bounded.
\end{proof}

While \gls{iIOS} represents a condition that is difficult to verify in practice, we observe that the resulting stability coefficient in Proposition~\ref{prop:iIOS_uniformstab} captures the data perturbation effect producing an exogenous disturbance, $u'_k=\hat\Psi(x'_k,\mc S^j_N)-\hat\Psi(x'_k,\mc S_N)=\hat\Phi(Cx'_k,\mc S^j_N)-\hat\Phi(Cx'_k,\mc S_N)$, through the evolution of two compatible trajectories.  The combination of these insights motivates us to resort to a more tractable dissipativity-based approach, which will allow us to directly control the sensitivity of the dynamics in \eqref{eq:dynamics} \gls{wrt} data perturbations, and hence the uniform stability of the resulting learning-enabled dynamics.

\section{Dissipativity-based certification of algorithmic stability}\label{sec:sensitivity_approach}

Inspired by the previous discussion, we will explore next how algorithmic stability shall not be considered just a static property of learning-enabled dynamics, but rather the output of a \emph{data sensitivity system}. We will make this interpretation explicit by treating the effect of replacing one training sample as an exogenous disturbance acting on a sensitivity system, and show that the resulting stability coefficient decomposes into a data-injection term and a dynamical sensitivity gain. 

\subsection{Sensitivity dynamics and uniform stability certificates}\label{subsec:main_results}
Given neighboring datasets $\mc S_N$ and $\mc S_N^j$, let us then consider the two trajectories $(x_k,y_k)$ and $(x_k',y_k')$ generated by \eqref{eq:dynamics} and
\[
	\left\{
	\begin{aligned}
	x'_{k+1}&=Ax'_k+B\hat\Phi(y'_k, \mc S_N^j),\\
	y'_k&=Cx'_k+D\hat\Phi(y'_k, \mc S_N^j),
	\end{aligned}
	\right.
\]
respectively, with the same initial condition $x_0=x'_0$. By defining $\Delta x_k \eqdef x_k-x'_k$ and $\Delta y_k \eqdef y_k-y'_k$, one can immediately decompose the input difference as follows:
\[
\begin{aligned}
	\Delta u_k
	&\eqdef\hat\Phi(y_k, \mc S_N)-\hat\Phi(y'_k, \mc S_N^j)\\
	&=\underbrace{\hat\Phi(y_k, \mc S_N)-\hat\Phi(y'_k, \mc S_N)}_{\reqdef \nu_k}
	+\underbrace{\hat\Phi(y'_k, \mc S_N)-\hat\Phi(y'_k, \mc S_N^j)}_{\reqdef w_k}.
\end{aligned}
\]
Here, $\nu_k$ is the sensitivity propagated through the same data-based mapping, whereas
$w_k$ is the new sensitivity injected by replacing one sample. Hence, the sensitivity
system reads as:
\begin{equation}\label{eq:sensitivity-system}
	\left\{
	\begin{aligned}
	\Delta x_{k+1}&=A\Delta x_k+B(\nu_k+w_k),\\
	\Delta y_k&=C\Delta x_k+D(\nu_k+w_k).
	\end{aligned}
	\right.
\end{equation}

Next, we make suitable assumptions on the mapping $\hat\Phi$:
\begin{standing}\label{standing:learning_mapping}
	For any $N\ge1$ and $\mc S_N\in\Xi^N$, the following conditions hold true:
	\begin{enumerate}
		\item[\normaltext{(i)}] $y\mapsto\hat\Phi(y,\mc S_N)$ satisfies an incremental \gls{IQC}, i.e., there exists a multiplier $\Pi\in\bbS^{m+p}$ such that
		\begin{equation}\label{eq:oracle-iqc}
			\begin{bmatrix}
				\Delta y\\
				\nu
			\end{bmatrix}^\top \Pi \begin{bmatrix}
				\Delta y\\
				\nu
			\end{bmatrix} \le 0,
		\end{equation}
		for all admissible \gls{IO} pairs satisfying $\Delta y=y-y'\in\R^p$, $\nu=\hat\Phi(y,\mc S_N)-\hat\Phi(y',\mc S_N)\in\R^m$;
		\item[\normaltext{(ii)}] $\exists \zeta_N>0$ s.t., for any $j\in\{1,\ldots,N\}$ and $y\in\R^p$, $\|\hat\Phi(y,\mc S_N)-\hat\Phi(y,\mc S^j_N)\|\le \zeta_N$.
	\end{enumerate}
	\hfill$\square$
\end{standing}

The first requirement is a standard assumption in dissipativity-based analysis \cite{lessard2016analysis,lessard2022analysis} that will be crucial to let the sensitivity dynamics be dissipative \gls{wrt} the supply rate defined by $\Pi$, a key factor to control the algorithmic stability. The term $w_k$, instead, is merely treated as an exogenous data perturbation channel. The second condition depends on how the mapping $\hat\Phi$ is obtained, or manipulates the samples directly. Remarkably, $\zeta_N$ happens to be proportional to $1/N$ in several relevant cases, making the generalization bound in Lemma~\ref{lemma:exp_bound} tight---see Appendix~\ref{app:1overN} for suitable examples and main derivations. We note that the proofs below also apply to the local case in which the neighboring trajectories evolve within some local set $\mathcal Y\subseteq\mathbb R^p$, and the developed bounds will hold true with:
\[
\zeta_N=\zeta_N(\mathcal Y)
\eqdef
\underset{j\in\{1,\ldots,N\}}{\textrm{sup}}
\underset{y\in\mathcal Y}{\textrm{sup}}~\|\hat\Phi(y,\mathcal S_N)-\hat\Phi(y,\mathcal S_N^j)\|.
\]

\begin{standing}\label{standing:loss_function}
	The loss function $H_N\mapsto\ell(H_N,\xi)$ is $\kappa$-Lipschitz continuous and so that $0\le\ell(H_N,\xi)\le M$.
	\hfill$\square$
\end{standing}

Standing Assumption~\ref{standing:loss_function} holds, for instance, for absolute, Huber-type or clipped losses, and squared loss functions restricted to compact subsets of the hypothesis space, $\mc X\subset\R^n$. 

\begin{theorem} \label{thm:induced-gain}
Let $(\Delta x_k,\Delta y_k,\nu_k,w_k)$
satisfy the sensitivity dynamics \eqref{eq:sensitivity-system} and the \gls{IQC} \eqref{eq:oracle-iqc}. Define
\[
\mathscr A \eqdef \begin{bmatrix} A & B & B \end{bmatrix},\qquad
\mathscr E \eqdef
\begin{bmatrix}
C & D & D\\
0 & I & 0
\end{bmatrix}.
\]
Suppose there exist $P\succ0$, $\rho$, $\gamma_1$, $\gamma_2\ge0$ such that
\ifTwoColumn
\else
	\begin{equation}\label{eq:main-lmi-induced-gain}
	\mathscr M(P,\rho,\gamma_1,\gamma_2)\eqdef
	\mathscr A^\top P\mathscr A
	-
	\diag(\rho^2P,0,0)
	-
	\gamma_1\mathscr E^\top\Pi\mathscr E
	-
	\diag(0,0,\gamma_2 I)
	\preccurlyeq0,
\end{equation}
\fi
Then, for every $T\ge1$,
\begin{equation}\label{eq:induced-gain-bound}
	\|\Delta x_T\|
	\le
	\sqrt{\frac{\gamma_2}{\lambda_{\normaltext{\textrm{min}}}(P)}}
	\left(
	\sum_{k=0}^{T-1}\rho^{2(T-k-1)}\|w_k\|^2
	\right)^{1/2}.
\end{equation}
\hfill$\square$
\end{theorem}
\begin{proof}
	Given a set of samples $\mc S_N$, Standing Assumption~\ref{standing:system_dynamics} guarantees that the dynamics in \eqref{eq:dynamics} is well-posed, and hence for any $x_0$ and $\mc S_N$ there exists a unique trajectory $\{x_k\}_{k=0}^{T}$. The same argument can be likewise applied to $\mc S^j_N$ to eventually establish the well-posedness of the sensitivity dynamics in \eqref{eq:sensitivity-system}.

	Then, let us define a storage function $V(\Delta x)\eqdef\Delta x^\top P\Delta x$, along with $s_k\eqdef\operatorname{col}(\Delta x_k,\nu_k,w_k)$. With matrices $\mathscr A$ and $\mathscr E$ defined in the main statement, one readily obtains:
	\[
		\Delta x_{k+1}=\mathscr A s_k,
		\qquad
		\begin{bmatrix}\Delta y_k\\ \nu_k\end{bmatrix}=\mathscr E s_k.
	\]
	With the introduced notation, it is thus immediate to verify that the \gls{MI} in \eqref{eq:main-lmi-induced-gain} is equivalent to the pointwise dissipativity inequality:
	\[
		V(\Delta x_{k+1})-\rho^2V(\Delta x_k)
		\le
		\gamma_1
		\begin{bmatrix}\Delta y_k\\ \nu_k\end{bmatrix}^{\!\top}
		\Pi
		\begin{bmatrix}\Delta y_k\\ \nu_k\end{bmatrix}
		+\gamma_2\|w_k\|^2.
	\]
	Along admissible operator $\hat\Phi$ trajectories, the \gls{IQC} term is nonpositive by
	\eqref{eq:oracle-iqc} in Standing Assumption~\ref{standing:learning_mapping}.(i). Since $\gamma_1\ge0$, we immediately obtain:
	\[
		V(\Delta x_{k+1})\le \rho^2V(\Delta x_k)+\gamma_2\|w_k\|^2.
	\]
	To meet Definition~\ref{def:unif_stab}, the two trajectories have the same initial condition by construction, so $\Delta x_0=0$, which implies $V(\Delta x_0)=0$. Telescoping the relation above leaves us with:
	\[
		V(\Delta x_T)
		\le
		\gamma_2
		\sum_{k=0}^{T-1}\rho^{2(T-k-1)}\|w_k\|^2.
	\]
	Finally, by noticing that $V(\Delta x_T)\ge\lambda_{\textrm{min}}(P)\|\Delta x_T\|^2$, we directly obtain
	\eqref{eq:induced-gain-bound}.
\end{proof}

\begin{remark}
	In case the multiplier class in \eqref{eq:oracle-iqc} is a cone, the positive scalar $\gamma_1$ can be absorbed in $\Pi$ directly and one may set $\gamma_1=1$ without loss of generality. If $\Pi$ is fixed or normalized, however, $\gamma_1\ge0$ remains a decision variable.
	\hfill$\square$
\end{remark}

The robust performance object is hence the induced gain from $w$ to the terminal state sensitivity $\Delta x_T$. The result above isolates this gain before imposing any specific pointwise bound on $w_k$. It is useful to introduce the weighted data-injection norm $\|w\|_{\rho,T}
	\eqdef
	\left(
	\sum_{k=0}^{T-1}\rho^{2(T-1-k)}\|w_k\|^2
	\right)^{1/2}$, which turns the bound in Theorem~\ref{thm:induced-gain} as
$
	\|\Delta x_T\|
	\le
	\sqrt{\gamma_2/\lambda_{\textrm{min}}(P)}\,
	\|w\|_{\rho,T}.
$

A uniform stability bound is recovered by bounding the data-injection signal uniformly, yielding the following corollary:

\begin{corollary}\label{cor:uniform-stability-finite-gain}
Under the same hypotheses of Theorem~\ref{thm:induced-gain}, and in case \eqref{eq:main-lmi-induced-gain} is solved for $\rho\in[0,1)$, the algorithm
$\mc A_N(\mc S_N)=x_T$ associated with the dynamics in \eqref{eq:dynamics} is uniformly stable with
\begin{equation}\label{eq:eta-induced-gain}
	\eta
	\le
	\kappa \zeta_N
	\sqrt{\frac{\gamma_2}{\lambda_{\normaltext{\textrm{min}}}(P)(1-\rho^2)}}.
	\end{equation}
\hfill$\square$
\end{corollary}
\begin{proof}
	By Standing Assumption~\ref{standing:learning_mapping}.(ii), $\|w_k\|\le \zeta_N$ for all $k\ge0$. Since $\rho<1$,
	\[
		\|w\|_{\rho,T}
		\le
		\zeta_N
		\left(\sum_{k=0}^{T-1}\rho^{2(T-k-1)}\right)^{1/2}
		\le
		\frac{\zeta_N}{\sqrt{1-\rho^2}}.
	\]
	With $\mc A_N(\mc S_N^j)=x'_T$, for any $j\in\{1,\ldots,N\}$, from Theorem~\ref{thm:induced-gain} one directly obtains:
	\[
		\|x_T-x'_T\|
		\le
		\zeta_N
		\sqrt{\frac{\gamma_2}{\lambda_{\textrm{min}}(P)(1-\rho^2)}}.
	\]
	Using the Lipschitz property of the loss function in Standing Assumption~\ref{standing:loss_function},
	$
		\sup_{\xi\in\Xi}
		|\ell(x_T,\xi)-\ell(x'_T,\xi)|
		\le
		\kappa
		\|x_T-x'_T\|,
	$
	which immediately proves \eqref{eq:eta-induced-gain}.
\end{proof}

Our dissipativity-based perspective allows one to identify three distinct terms contributing to the stability coefficient, i.e.,
\[
\eta
	\le
	\underbrace{\kappa}_{\text{loss readout}}
	\underbrace{\zeta_N}_{\text{data injection}}
	\underbrace{
	\sqrt{\frac{\gamma_2}{\lambda_{\textrm{min}}(P)(1-\rho^2)}}
	}_{\text{dynamical sensitivity gain}}.
\]
It turns out that algorithmic stability is directly controlled by how well the sensitivity dynamics either dissipates or amplifies the new data sensitivity injected at each step. 
Then, to obtain tight generalization bounds consistent with $N$, e.g., Lemma~\ref{lemma:exp_bound}, notice that one must have $\zeta_N\propto 1/N$, which happens to be the case for several learning mappings of interest---see Appendix~\ref{app:1overN}.

We conclude by stating the following result that will be useful for the momentum-based schemes discussed in \S \ref{sec:data_based_algorithms}:

\begin{corollary}\label{cor:Rass}
Let the hypothesis be $H_N=Rx_T$ for some matrix $R\in\R^{q\times n}$, and suppose the assumptions of
Theorem~\ref{thm:induced-gain} hold true. If $P\succeq R^\top R$ 
, then
\[
\sup_{\xi\in\Xi}
|\ell(Rx_T,\xi)-\ell(Rx_T',\xi)|
\le
\kappa\zeta_N
\sqrt{
\gamma_2
\frac{1-\rho^{2T}}{1-\rho^2}
},
\]
which in case $\rho\in[0,1)$ turns into:
\[
\sup_{\xi\in\Xi}
|\ell(Rx_T,\xi)-\ell(Rx_T',\xi)|
\le
\kappa\zeta_N
\sqrt{\frac{\gamma_2}{1-\rho^2}}.
\]
\hfill$\square$
\end{corollary}


\subsection{Robust performance interpretation}\label{subsec:robust_performance}
\begin{figure}[t]
	\centering
	\begin{tikzpicture}[>=latex]

	\node[draw, rectangle, minimum width=3.0cm, minimum height=1.0cm, align=center]
			(G) at (0, 0) {Linear system \\ $(A,B,C,D)$};

	\node[draw, rectangle, minimum width=3.0cm, minimum height=1.0cm, align=center]
			(Delta) at (0,-2) {$\hat\Phi$-increment \\ (IQC: $\Pi$)};

	\draw[->] (-3.2, 0.25) -- node[above]{\small $w_k$} (G.west |- {(0,0.25)});

	\draw[->] (G.east |- {(0,0.25)}) -- node[above]{\small $\Delta z_T$} (3.2, 0.25);

	\draw[->] (G.east |- {(0,-0.25)}) -- (2.4,-0.25)
				-- (2.4,-2) -- node[above]{\small $\Delta y_k$} (Delta.east);

	\draw[->] (Delta.west) -- (-2.4,-2)
				-- (-2.4,-0.25) -- node[above]{\small $\nu_k$} (G.west |- {(0,-0.25)});

	\end{tikzpicture}
	\caption{Feedback interconnection of the sensitivity
	system~\eqref{eq:sensitivity-system}. 
	Uniform stability is a bounded-gain condition on $w_k \mapsto \Delta z_T$.}
	\label{fig:block_diagram}
\end{figure}
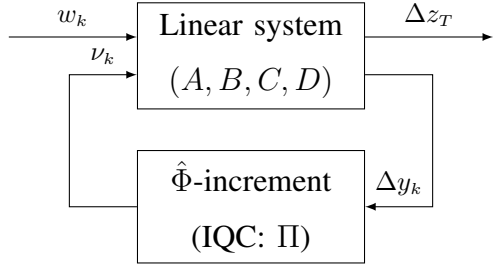

In our analysis the object of central interest coincides with the finite-horizon \gls{IO} mapping:
\begin{equation}
  \label{eq:perf_map}
  w \mapsto \Delta z_T,
  \qquad
  \Delta z_T(\xi) \eqdef \ell(x_T,\xi)-\ell(x_T',\xi),
\end{equation}
where $w_k = \hat\Phi(y_k', \mc S_N)-\hat\Phi(y_k', \mc S_N^j)$ is the exogenous disturbance introduced by replacing one training sample, and $\Delta z_T$ is the resulting terminal loss difference. Uniform algorithmic stability requires $\sup_{\xi\in\Xi}|\Delta z_T(\xi)|\leq\eta$ for every neighbouring dataset pair $\mc S_N$ and $\mc S_N^j$. Equivalently, this can be interpreted as a bounded-gain condition on the mapping in \eqref{eq:perf_map}.

This is precisely the structure of a standard robust-performance problem in the sense of $\mathcal{H}_\infty$/\gls{IQC} theory \cite{lessard2016analysis,megretski1997system}. Here, a linear system is driven by an exogenous disturbance $w$, with a structured uncertainty $\nu_k = \hat\Phi(y_k, \mc S_N)-\hat\Phi(y_k', \mc S_N)$ in the feedback connection, and a performance output $\Delta z_T$ whose induced gain must be certified. The dissipativity \gls{MI} in \eqref{eq:main-lmi-induced-gain} is the corresponding robust performance certificate.

Figure~\ref{fig:block_diagram} illustrates the feedback interconnection of the sensitivity system~\eqref{eq:sensitivity-system}, which consists of the following blocks:
\begin{enumerate}
  \item[i)] Linear system: The realization $(A,B,C,D)$ governs how sensitivity $\Delta x_k$ propagates across iteration;

  \item[ii)] Structured uncertainty: The increment $\nu_k = \hat\Phi(y_k,\mc S_N)-\hat\Phi(y_k', \mc S_N)$, constrained by the \gls{IQC}~\eqref{eq:oracle-iqc} with multiplier $\Pi$, constitutes the feedback nonlinearity;

  \item[iii)] Exogenous perturbation: The signal $w_k = \hat\Phi(y_k', \mc S_N)-\hat\Phi(y_k', \mc S_N^j)$ is the data-injection channel, carrying the sensitivity introduced by replacing sample $\xi^{(j)}$ at every $k$.
\end{enumerate}

Two design parameters directly shape performance in $\mathscr M$ in \eqref{eq:main-lmi-induced-gain}. Specifically, the multiplier $\Pi$ encodes different prior knowledge about the learning mapping, while the rate $\rho$ controls how aggressively past sensitivity is discounted. Note that different choices of the $\Pi$ may reduce conservatism (e.g., Zames-Falb conditions) at the cost of a higher-dimensional \gls{MI} compared to slope-restricted/sector constraints, cocoercivity, or monotonicity, all corresponding to a recognized class of \gls{IQC} \cite{yakubovich1971procedure,scherer2022dissipativity}. The resulting sensitivity gain is then a function of $(P,\rho,\gamma_1,\gamma_2)$, which can be optimized over to obtain the tightest bound for the chosen multiplier class $\Pi$---see \S \ref{subsec:computational_aspects}. 
Notice that the same argument used in Theorem~\ref{thm:induced-gain}, specific to static pointwise incremental \glspl{IQC}, also applies to the augmented sensitivity system including the state of the multiplier filter, and dynamic hard \glspl{IQC} over $T$ can be incorporated directly.


\subsection{Optimizing the dynamical sensitivity gain}\label{subsec:computational_aspects}
The adopted dissipativity-based perspective therefore offers the opportunity to directly optimize over the sensitivity gain, since it merely depends on quantities related with the storage function and supply rate of the exogenous disturbance $w_k$. For a feasible tuple $(P,\rho,\gamma_1,\gamma_2)$ of \eqref{eq:main-lmi-induced-gain}, define the sensitivity gain:
\[
 G(P,\rho,\gamma_2)
 \eqdef
 \sqrt{\frac{\gamma_2}{\lambda_{\textrm{min}}(P)(1-\rho^2)}}.
\]
Since the \gls{MI} is homogeneous in $(P,\gamma_1,\gamma_2)$, one can then impose a normalization such as $P\succcurlyeq I$.
For a fixed $\rho\in[0,1)$, let us then define
\begin{equation}\label{eq:gamma-star-rho}
	\begin{aligned}
	\gamma_2^\star(\rho)\eqdef
	&\underset{P,\gamma_1,\gamma_2}{\textrm{min}}&& \gamma_2\\
	&~\textrm{ s.t. }&&\mathscr M(P,\rho,\gamma_1,\gamma_2)
	\preccurlyeq0,\\
	&&& P\succcurlyeq I, \gamma_1, \gamma_2\ge0.
	\end{aligned}
\end{equation}
Note that \eqref{eq:gamma-star-rho} amounts to an \gls{SDP} when $\rho$ is fixed. The optimized sensitivity gain then reads as:
\begin{equation}\label{eq:G-star}
	G^\star
	\eqdef
	\underset{\rho\in[0,1)}{\textrm{inf}}~
	\sqrt{\frac{\gamma_2^\star(\rho)}{1-\rho^2}},
\end{equation}
which on the other hand requires scalar search over $[0,1)$ to apply Corollaries \ref{cor:uniform-stability-finite-gain} or \ref{cor:Rass}. The resulting uniform stability coefficient is thus $\eta\le\kappa \zeta_N G^\star$, with optimal sensitivity gain $G^\star$ being the tightest bound the chosen multiplier class $\Pi$ can certify for the realization $(A,B,C,D)$ at hand.

Armed with this, the generalization bound in Lemma~\ref{lemma:exp_bound} yields, with probability $1-\delta$ over the multisample extraction:
\[
	r(x_T)
	-
	\hat r(x_T, \mc S_N)
	\le
	\kappa \zeta_N G^\star
	+(N\kappa \zeta_N G^\star+M)
	\sqrt{\frac{2\ln(1/\delta)}{N}},
\]
which quantifies the generalization capabilities of the observed learning-enabled dynamics, measured by the loss function $\ell$.
\section{Data-based optimization dynamics}\label{sec:data_based_algorithms}
We now apply the proposed dissipativity-based perspective to the analysis of the algorithmic stability of data-based optimization dynamics. 
Specifically, we consider here algorithms tailored to \gls{SAA} for \gls{ERM}, where the latter denotes the empirical counterpart of the stochastic optimization problem:
\[
	\underset{x\in\mathbb R^n}{\textrm{min}}~\expected{\prob{P}}{f(x,\xi)},
\]
with $f:\R^n\times\Xi\to\R$, whose sample-based approximation relying on \gls{iid} samples gives the resulting empirical risk:
\begin{equation}\label{eq:empirical_risk_2}
	\underset{x\in\mathbb R^n}{\textrm{min}}~\frac1N\sum_{i=1}^N f(x,\xi^{(i)}).
\end{equation}

\subsection{Gradient descent method}
Given a step-size $\alpha>0$, applying the \gls{GD} to \eqref{eq:empirical_risk_2} yields:
\[
	x_{k+1}=x_k-\alpha \left(\frac1N\sum_{i=1}^N \nabla f(x_k,\xi^{(i)})\right),
\]
which trivially amounts to a special case of the feedback form in \eqref{eq:dynamics} with system matrices:
\[
    A=I,
    \qquad
    B=-\alpha I,
    \qquad
    C=I,
    \qquad
    D=0,
\]
and data-driven operator $\hat\Phi(y,\mc S_N)=\frac1N\sum_{i=1}^N \nabla f(y,\xi^{(i)})$.
Thus, the general sensitivity system \eqref{eq:sensitivity-system} specializes to the sensitivity dynamics of \gls{GD}. Here, in view of the structure of the output matrix $C$, the data-injection disturbance generated by the replacement of one sample satisfies: 
\[
	w_k\!=\!\hat\Phi(x'_k, \mc S_N)-\hat\Phi(x'_k, \mc S_N^j)\!=\!\frac1N\left(\nabla f(x'_k,\xi^{(j)})-\nabla f(x'_k,\xi')\right).
\]
Therefore, according to Corollary~\ref{cor:SAA_mappings}, if there exists $B>0$ such that $\|\nabla f(x,\xi)\|\le B$ for all $x\in\mathbb R^n$ and $\xi\in\Xi$, then $\|w_k\|\le 2B/N$ for all $k\ge0$, i.e., $\zeta_N\le2B/N$ in the notation of Theorem~\ref{thm:induced-gain}. The remaining part of the stability estimate is purely dynamical, as it depends on how the
\gls{GD} dynamics dissipates the injected sensitivity. For completeness, we next state and prove a popular result available in the literature \cite{hardt2016train}:

\begin{corollary}\label{cor:gd-stability}
	Given any $N\ge1$ and $\mc S_N\in\Xi^N$, let $x\mapsto f(x,\xi)$ be
	$m$-strongly convex and $L$-smooth. Choose $\alpha\in(0,2/(L+m)]$, and define $q\eqdef\normaltext{\textrm{max}}\{|1-\alpha m|,|1-\alpha L|\}<1$
	Then, the \gls{GD} is uniformly stable with
	\begin{equation}\label{eq:gd-finite-time-stability}
		\eta
		\le
		\frac{2\alpha \kappa B}{N(1-q)}.
	\end{equation}
	If $\alpha\in(0,1/L)$, then $q=1-\alpha m$, and $\eta\le2 \kappa B/mN$.
	\hfill$\square$
\end{corollary}
\begin{proof}
	Pick any $j\in\{1,\ldots,N\}$, and let $\mc S_N^j$ be the associated neighboring dataset of $\mc S_N$. In addition, let $x_k$, $x'_k$ be the corresponding state trajectories of the \gls{GD} initialized from the same $x_0=x'_0$. By adding and subtracting $\frac1N(\sum_{i=1,i\neq j}^N \nabla f(x'_k,\xi^{(i)}) + \nabla f(x'_k,\xi'))$ to $\Delta x_k=x_k-x'_k$ one immediately obtains $\Delta x_{k+1}=\Delta x_k-\alpha (\nu_k + w_k)$, with exactly the same definitions for signals $\nu_k$ and $w_k$ given in \S \ref{sec:sensitivity_approach}.
	Since $x\mapsto\frac1N\sum_{i=1}^N f(x,\xi^{(i)})$ is $m$-strongly convex and $L$-smooth, the gradient step $\mc T_N(x)\eqdef x-\frac{\alpha}{N}\sum_{i=1}^N \nabla f(x,\xi^{(i)})$ is contractive with rate $q$. Therefore, one directly obtains:
	\ifTwoColumn
\else
	\[
		\begin{aligned}
			\|\Delta x_{k+1}\|
			&=\|\mc T_N(x_k)-\mc T_N^j(x'_k)\|\le \|\mc T_N(x_k)-\mc T_N(x'_k)\|
				+\|\mc T_N(x'_k)-\mc T_N^j(x'_k)\|\\
			&\le q\|\Delta x_k\|+\alpha\|w_k\|.
		\end{aligned}
	\]
	\fi
	Moreover, $w_k=\frac1N\left(\nabla f(x'_k,\xi^{(j)})-\nabla f(x'_k,\xi')\right)$, so $\|w_k\|\le 2B/N$. Hence $\|\Delta x_{k+1}\|\le q\|\Delta x_k\|+2\alpha B/N$. Since $\Delta x_0=0$, unrolling the recursion over $T$ immediately leaves us with
	\[
		\|\Delta x_T\|
		\le
		\frac{2\alpha B}{N}\sum_{k=0}^{T-1}q^k
		=
		\frac{2\alpha B}{N}\frac{1-q^T}{1-q}\le\frac{2\alpha B}{N(1-q)}.
	\]
	Using the $\kappa$-Lipschitz continuity of the loss function proves \eqref{eq:gd-finite-time-stability}, while the remaining claim follows by simply considering that for $\alpha\in(0,1/L)$, $q=1-\alpha m$.
\end{proof}

The bounds in Corollary~\ref{cor:gd-stability} admit the same interpretation as the general results in \S \ref{subsec:main_results}. Specifically, the replacement of one sample injects at most $\alpha\|w_k\|\le 2\alpha B/N$ units of state sensitivity at each iteration. The contraction factor $q$ determines how much of this injected sensitivity remains after consecutive iterations. The stability coefficient is therefore the product of three factors, which in the strongly convex case generalizes because it forgets one-sample perturbations geometrically, and the accumulated sensitivity saturates. If, instead, $q=1$, one recovers the linear-in-time estimate $\eta\le 2\kappa\alpha B T/N$. Finally, with $q>1$, sample perturbations may be amplified.

Notice that if $x\mapsto\frac1N\sum_{i=1}^N f(x,\xi^{(i)})$ is $m$-strongly convex and $L$-smooth, then the map
$x\mapsto\frac1N\sum_{i=1}^N \nabla f(x,\xi^{(i)})$ is slope-restricted in the sector $[m,L]$. In view of the structure of the output matrix $C=I$, the pair $(\Delta x_k,\nu_k)$ then satisfies the \gls{IQC}:
\begin{equation}
    \begin{bmatrix}
        \Delta x_k\\
        \nu_k
    \end{bmatrix}^{\!\top}
    \begin{bmatrix}
m L & -(L+m)/2 \\
-(L+m)/2 & 1
	\end{bmatrix}		
    \begin{bmatrix}
        \Delta x_k\\
        \nu_k
    \end{bmatrix}
    \le 0.
\end{equation}
Therefore, the general \gls{MI}-based certificate in \eqref{eq:main-lmi-induced-gain} applies directly. In this simple case, the resulting sensitivity gain is consistent with the explicit contraction calculation above:
the dissipativity-based theorem recovers the classical contractivity-based uniform stability bound while assigning a systems-theoretic meaning to each factor.
\gls{GD} is perhaps the simplest instance illustrating the main message conveyed by this paper: algorithms that dissipate data sensitivity admit stability-based generalization guarantees, while algorithms that accumulate sensitivity may exhibit poor generalization.

\subsection{Heavy-ball algorithm}
Given strictly positive step-size $\alpha$ and momentum term $\beta$, the \gls{HB} iteration applied to the minimization of \eqref{eq:empirical_risk_2} reads as:
\[
 x_{k+1}=x_k-\alpha\left(\frac{1}{N}\sum_{i=1}^N\nabla f(x_k,\xi^{(i)})\right)+\beta(x_k-x_{k-1}).
\]
With the lifted state $\tilde x_k\eqdef\col(x_k,x_{k-1})$, the resulting dynamics can be written in control-affine form as \cite{lessard2016analysis}:
\[
\left\{
\begin{aligned}
	\tilde x_{k+1}&=A_{\rm HB}\tilde x_k+B_{\rm HB}\hat\Phi(y_k,\mc S_N),\\
	y_k&=C_{\rm HB}\tilde x_k,
\end{aligned}
\right.
\]
with nonlinear map $\hat\Phi(y_k,\mc S_N)=\frac{1}{N}\sum_{i=1}^N\nabla f(y_k,\xi^{(i)})$, and
\ifTwoColumn
\else
	\[
		A_{\rm HB}
		=
		\begin{bmatrix}
		(1+\beta)I&-\beta I\\
		I&0
		\end{bmatrix},
		\quad
		B_{\rm HB}
		=
		\begin{bmatrix}
		-\alpha I\\
		0
		\end{bmatrix},
		\quad
		C_{\rm HB}=\begin{bmatrix}I&0\end{bmatrix},
		\quad
		D_{\rm HB}=0.
	\]
\fi
As in the \gls{GD} case, the neighboring-data perturbation associated to the sensitivity system amounts to $w_k=\tfrac{1}{N}(\nabla f(x_k',\xi^{(j)})-\nabla f(x_k',\xi'))$, and if $\|\nabla f(x,\xi)\|\le B$ for all $x\in\R^n$, $\xi\in\Xi$, one obtains $\|w_k\|\le 2B/N$ for all $k\ge0$. Hence, Corollary~\ref{cor:Rass} applied to hypothesis $H_N=[I \ 0]\tilde x_T$ leaves us with:
\[
 \eta_{\rm HB}
 \le
 \frac{2\kappa B}{N}G_{\rm HB}^\star,
\]
where $G_{\rm HB}^\star$ follows by solving \eqref{eq:gamma-star-rho}--\eqref{eq:G-star} with
$(A,B,C,D)=(A_{\rm HB},B_{\rm HB},C_{\rm HB},D_{\rm HB})$.

Akin to the proof of Corollary~\ref{cor:gd-stability}, we note that standard algorithmic stability arguments would implicitly require a pointwise one-step estimate in a prescribed norm as, e.g.,
\[
    \|\Delta x_{k+1}\|
    \leq q\|\Delta x_k\| + b\|w_k\|, \text{ with } q<1.
\]
Such a condition immediately implies a finite-gain bound, but it is stronger. The robust performance certificate of Theorem~\ref{thm:induced-gain} only requires the existence of a storage function for which the closed-loop sensitivity interconnection dissipates the data injection supply. This distinction becomes crucial for algorithms such as \gls{HB}. In fact, after lifting the state the dynamics may be Schur stable and have finite \gls{IO} gain even though the lifted update is not a contraction in the Euclidean norm.
\begin{example}\label{example:HB}
	Let $f(x,\xi)=\frac12\varrho(\xi)x^2$, with $0<m\leq \varrho(\xi)\leq L$ for all $\xi\in\Xi$, and define the empirical curvature as $\hat\varrho_N\eqdef\frac1N\sum_{i=1}^N \varrho(\xi^{(i)})$.
	Then the empirical gradient mapping coincides with $\hat\Phi(x,\mc S_N)=\hat\varrho_N x$. For a neighboring dataset $\mc S_N^j$, let $\hat\varrho_N^j$ denote the corresponding empirical curvature. In this quadratic, strongly convex setting, the \gls{HB} dynamics can be recasted into the linear lifted sensitivity dynamics by means of matrices:
	\[
		A_{\rm HB}(\hat\varrho_N)
		=
		\begin{bmatrix}
		1+\beta-\alpha\hat\varrho_N & -\beta\\
		1 & 0
		\end{bmatrix},
		\qquad
		B_{\rm HB}
		=
		\begin{bmatrix}
		-\alpha\\
		0
		\end{bmatrix},
	\]
	and
	$
	\begin{aligned}
		w_k
		=
		(\hat\varrho_N-\hat\varrho_N^j)x_k'
		=
		(\varrho(\xi^{(j)})-\varrho(\xi')x_k'/N.
	\end{aligned}
	$
	In particular, if $|x_k'|\leq R$ for all $k\ge0$, then $|w_k|\leq (L-m)R/N$. However, note that there exist $(\alpha,\beta)$ for which $A_{\rm HB}(\hat\varrho_N)$ is Schur stable but not a strict Euclidean contraction, e.g., $\beta=1/4$ and $\alpha=3/4\hat\varrho_N$, yielding
	\[
		A_{\rm HB}
		=
		\begin{bmatrix}
		1/2 & -1/4\\
		1 & 0
		\end{bmatrix}.
	\]
	The spectral radius is $\textnormal{\textrm{max}}~\{|\lambda_1(A_{\rm HB})|,|\lambda_2(A_{\rm HB})|\}=1/2<1$, while $\|A_{\rm HB}\|>1$. Hence, the \gls{HB} sensitivity dynamics has finite data sensitivity gain although the contractivity argument fails. In view of the Schur stability of $A_{\rm HB}$, standard Lyapunov arguments guarantee the existence of some $P\succ0$ satisfying $A_{\rm HB}^\top P A_{\rm HB}-P=-Q$, for some $Q\succ0$. Thus, for any $\rho\in(1/2,1)$, one can find $P\succ0$ such that $A_{\rm HB}^\top P A_{\rm HB}-\rho^2P\prec0$, giving the dissipativity certificate and hence a finite gain.
	\hfill$\square$
\end{example}

\subsection{Nesterov’s accelerated gradient method}
For given positive step-size $\alpha$ and momentum term $\beta$, the \gls{NAG} iteration applied to the minimization of \eqref{eq:empirical_risk_2} reads as:
\[
\left\{
\begin{aligned}
	z_k&=x_k+\beta(x_k-x_{k-1}),\\
	x_{k+1}&=z_k-\frac{\alpha}{N}\sum_{i=1}^N\nabla f(z_k,\xi^{(i)}).
\end{aligned}
\right.
\]
With the lifted state $\tilde x_k$ as for \gls{HB}, we readily obtain a linear system representation with nonlinear input and system matrices:
\ifTwoColumn
\else
	\[
		A_{\rm NAG}
		=
		\begin{bmatrix}
		(1+\beta)I&-\beta I\\
		I&0
		\end{bmatrix},
		\quad
		B_{\rm NAG}
		=
		\begin{bmatrix}
		-\alpha I\\
		0
		\end{bmatrix},
		\quad
		C_{\rm NAG}=\begin{bmatrix}(1+\beta)I&-\beta I\end{bmatrix},
		\quad
		D_{\rm NAG}=0.
	\]
\fi
Akin to the \gls{GD} and \gls{HB} methods, one can bound $\|w_k\|\le 2B/N$ for all $k\ge0$, and obtain a uniform stability bound once considered the hypothesis $H_N=[I \quad 0]\tilde x_T$ as:
\[
 \eta_{\rm NAG}
 \le
 \frac{2\kappa B}{N}G_{\rm NAG}^\star,
\]
where $G_{\rm NAG}^\star$ is obtained by solving \eqref{eq:gamma-star-rho}--\eqref{eq:G-star} with
$(A,B,C,D)=(A_{\rm NAG},B_{\rm NAG},C_{\rm NAG},D_{\rm NAG})$.
The example provided next confirms that, also for the lifted \gls{NAG} dynamics where contractivity may not hold, our dissipativity-based angle allows one to optimize over the dynamical sensitivity gain:
\begin{example}\label{example:NAG}
	With the same ingredients as in Example~\ref{example:HB}, the (lifted) \gls{NAG} iteration is characterized by 
	\[
	\begin{aligned}
		A_{\rm NAG}(\hat\varrho_N)
		&=
		\begin{bmatrix}
		(1-\alpha\hat\varrho_N)(1+\beta)
		&
		-(1-\alpha\hat\varrho_N)\beta\\
		1 & 0
		\end{bmatrix},
	\end{aligned}
	\]
	$B_{\rm NAG}=B_{\rm HB}$, and $w_k=\frac{\varrho(\xi^{(j)})-\varrho(\xi')}{N}z_k'$.
	As before, if $|z_k'|\leq R$ for all $k\ge0$, then $|w_k|\leq (L-m)R/N$. Also in this case, one can find suitable choices of $(\alpha,\beta)$ for which $A_{\rm NAG}(\hat\varrho_N)$ is Schur stable but not a strict Euclidean contraction, e.g., $\beta=1/2$ and $\alpha=1/2\hat\varrho_N$, which yields
	a spectral radius $\textnormal{\textrm{max}}~\{|\lambda_1(A_{\rm NAG})|,|\lambda_2(A_{\rm NAG})|\}=1/2<1$, while $\|A_{\rm NAG}\|>1$, and the same conclusion for \gls{HB} also applies to \gls{NAG}.
	\hfill$\square$
\end{example}

\begin{remark}\label{remark:comparison}
	Theorem~\ref{thm:induced-gain}, together with Corollaries~\ref{cor:uniform-stability-finite-gain} and \ref{cor:Rass}, produces algorithm-specific stability bounds, i.e.,
	\[
	\eta_{\rm GD}\le \frac{2\kappa B}{N}G_{\rm GD}^\star,
	\;
	\eta_{\rm HB}\le \frac{2\kappa B}{N}G_{\rm HB}^\star,
	\;
	\eta_{\rm NAG}\le \frac{2\kappa B}{N}G_{\rm NAG}^\star.
	\]
	Thus the underlying framework allows one to compare iterative algorithms not only through convergence rate, but also through their generalization capabilities.
	\hfill$\square$
\end{remark}

\subsection{Numerical experiments}
\begin{figure}[t]
    \centering
	\ifTwoColumn
\else
		\includegraphics[width=.6\columnwidth]{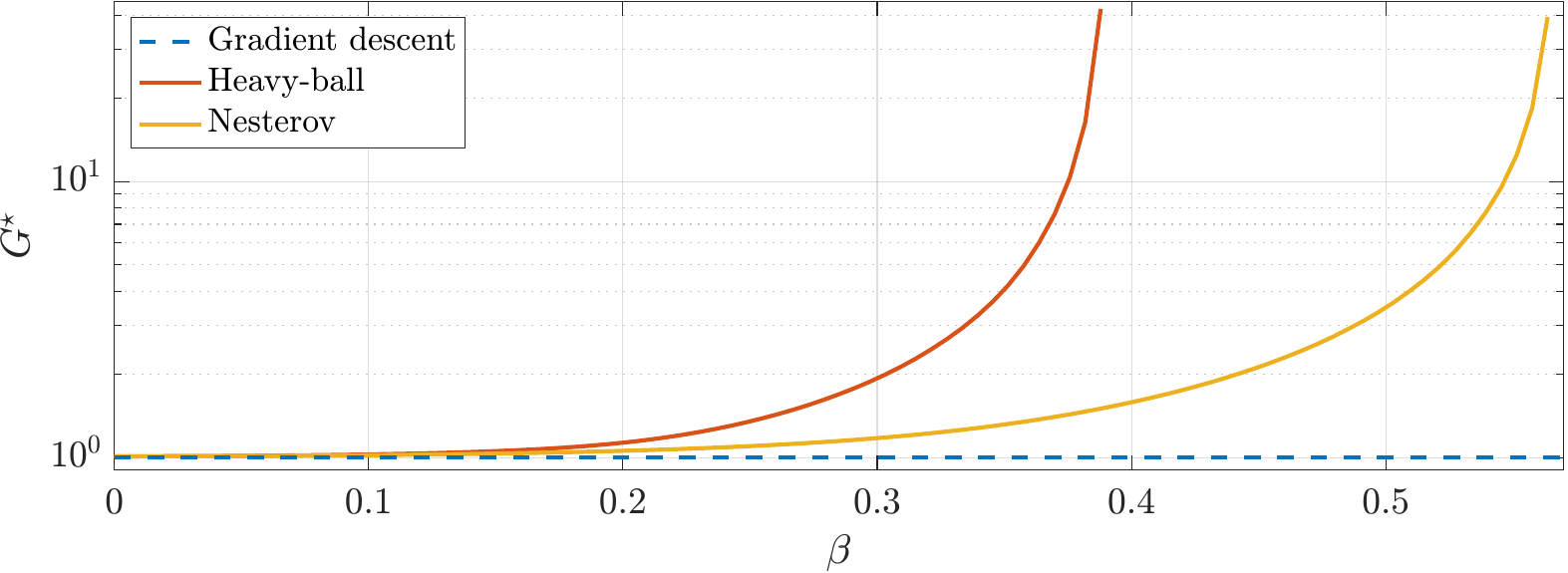}
		\includegraphics[width=.6\columnwidth]{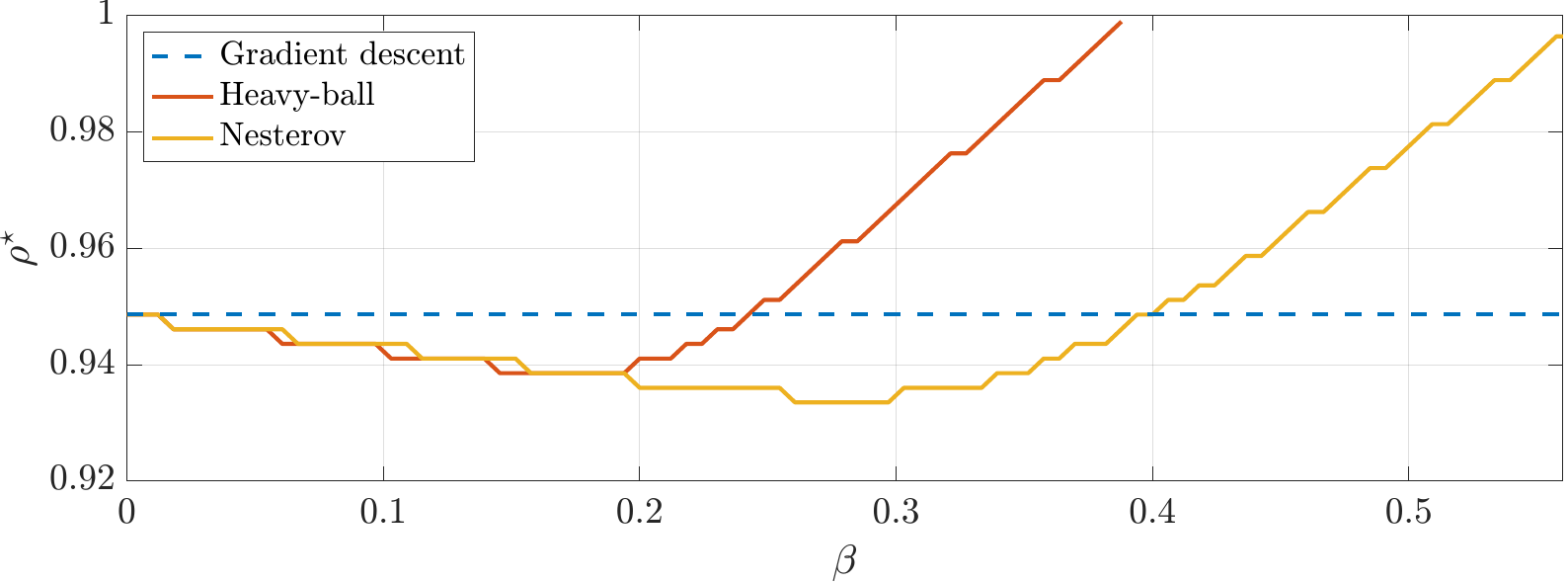}
	\fi
    \caption{Top: Certified data sensitivity gain (logarithmic scale). Bottom: Optimized discount factor as functions of the momentum parameter $\beta$. 
	Coloured lines are shown only over those values of $\beta$ for which the \gls{SDP} certificate is feasible.}
    \label{fig:sensitivity_gain_momentum}
\end{figure}

We start by evaluating the dynamical component of the stability coefficient. 
Inspired by Remark~\ref{remark:comparison}, we fix the class of nonlinearity for $\hat\Phi$ by imposing the same sector \gls{IQC} $[m,L]$ for all methods, i.e., \gls{GD}, \gls{HB}, and \gls{NAG}, and compare the certified gains $G^\star$ obtained by solving \eqref{eq:gamma-star-rho}-\eqref{eq:G-star}. 
The simulations are run in MATLAB with MOSEK \cite{mosek} as solver, on a laptop with an Apple M2 chip featuring an 8-core CPU and 16 Gb RAM. 

\begin{figure}[t]
    \centering
	\ifTwoColumn
\else
		\includegraphics[width=.6\columnwidth]{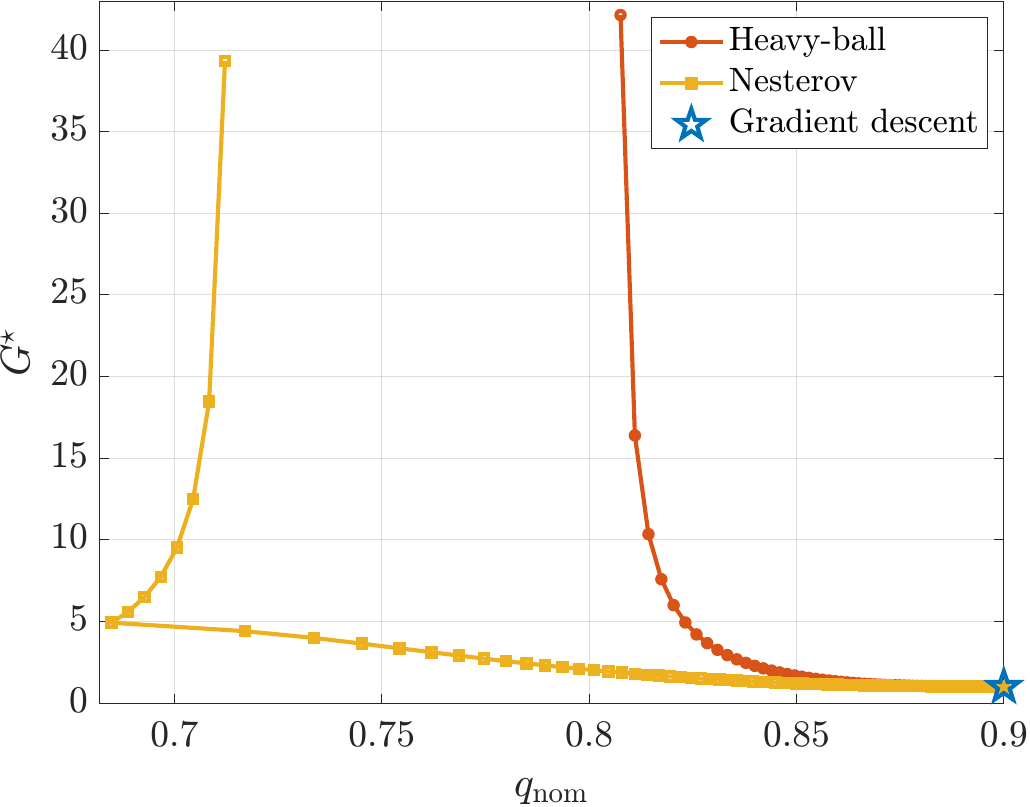}    
	\fi
	\caption{Convergence-sensitivity comparison. Each marker corresponds to a value of the momentum parameter $\beta$. 
	}
    \label{fig:tradeoff}
\end{figure}

In a scalar setting, i.e., $n=1$, by choosing $m=1$, $L=10$, and $\alpha=0.1$, Fig.~\ref{fig:sensitivity_gain_momentum} reports the certified sensitivity gain $G^\star$ and the corresponding optimized discount factor $\rho^\star$ as functions of the momentum parameter. Notice that, since the loss Lipschitz constant $\kappa$ and the data-injection size $\zeta_N$ are common across the methods, differences in the resulting stability certificates are entirely due to the algorithmic realization $(A,B,C,D)$.
For small momentum, the gains of \gls{HB} and \gls{NAG} are close to the \gls{GD} baseline. As $\beta$ increases, the certified gain grows rapidly, indicating that the algorithmic dynamics increasingly amplifies one-sample perturbations. In particular, \gls{HB} loses feasibility earlier than \gls{NAG}, which suggests a smaller certified robustness margin under the chosen quadratic storage and sector multiplier. The increase of $\rho^\star$ toward one further indicates that large momentum induces a longer memory of data perturbations.

Following the approach illustrated in Examples \ref{example:HB} and \ref{example:NAG}, 
we set $q_{\rm nom}\eqdef\textrm{max}_{t\in[m,L]} \textrm{max}_i |\lambda_i(A_{\rm alg}(t))|$ as the nominal worst-case convergence factor, where $A_{\rm alg}$ is the dynamical matrix obtained in the quadratic setting, ${\rm alg}\in\{$\gls{GD}, \gls{HB}, \gls{NAG}$\}$, and in Fig.~\ref{fig:tradeoff} we compare the certified data sensitivity gain $G^\star$ against it. 
While \gls{GD} appears as a single point, \gls{HB} and \gls{NAG} trace distinct curves as $\beta$ varies. Interestingly, nominal convergence and data sensitivity appear distinct properties, and methods with comparable values of $q_{\rm nom}$ may have substantially different certified gains. In particular, \gls{NAG} attains smaller nominal convergence factors than \gls{GD}, while maintaining moderate sensitivity over a wide range of $\beta$, whereas \gls{HB} exhibits a rapid increase in $G^\star$, as its certified robustness margin deteriorates. This is consistent with the robust performance interpretation provided in \S \ref{subsec:robust_performance}, where algorithmic stability amounts to a property of the data sensitivity system rather than as a consequence of nominal convergence alone.

\begin{figure}[t]
    \centering
	\ifTwoColumn
\else
		\includegraphics[width=.6\columnwidth]{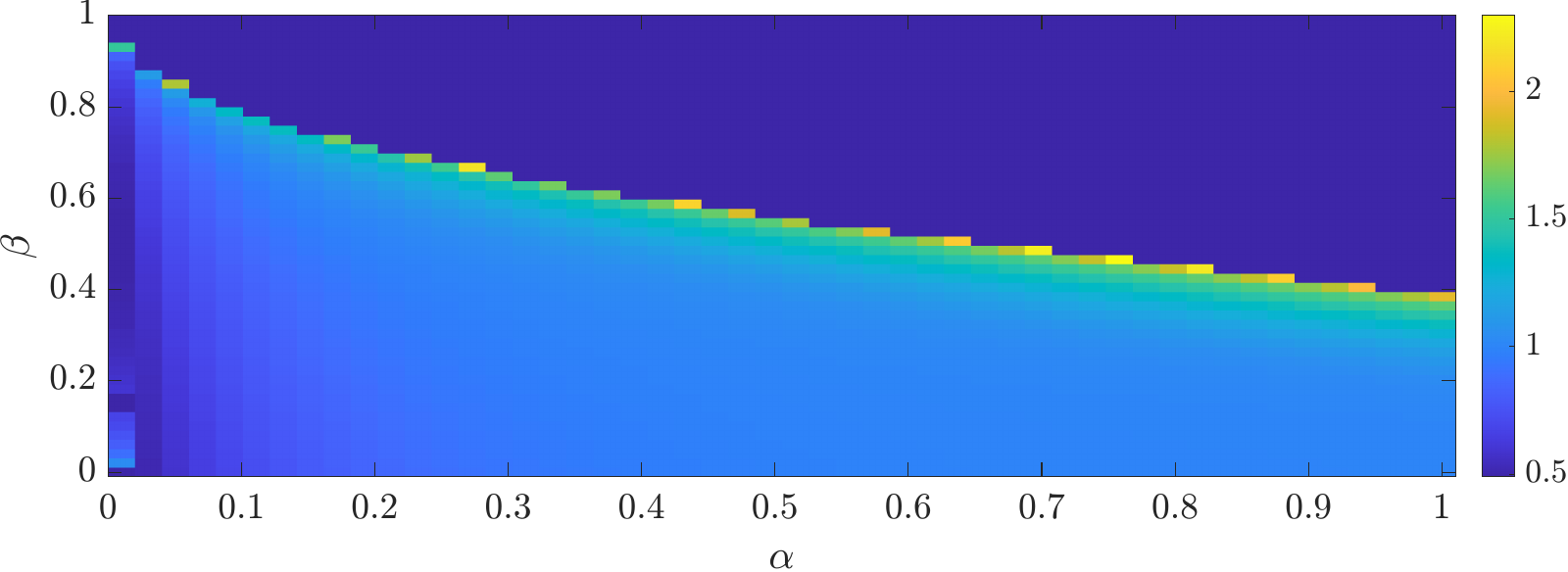}
    	\includegraphics[width=.6\columnwidth]{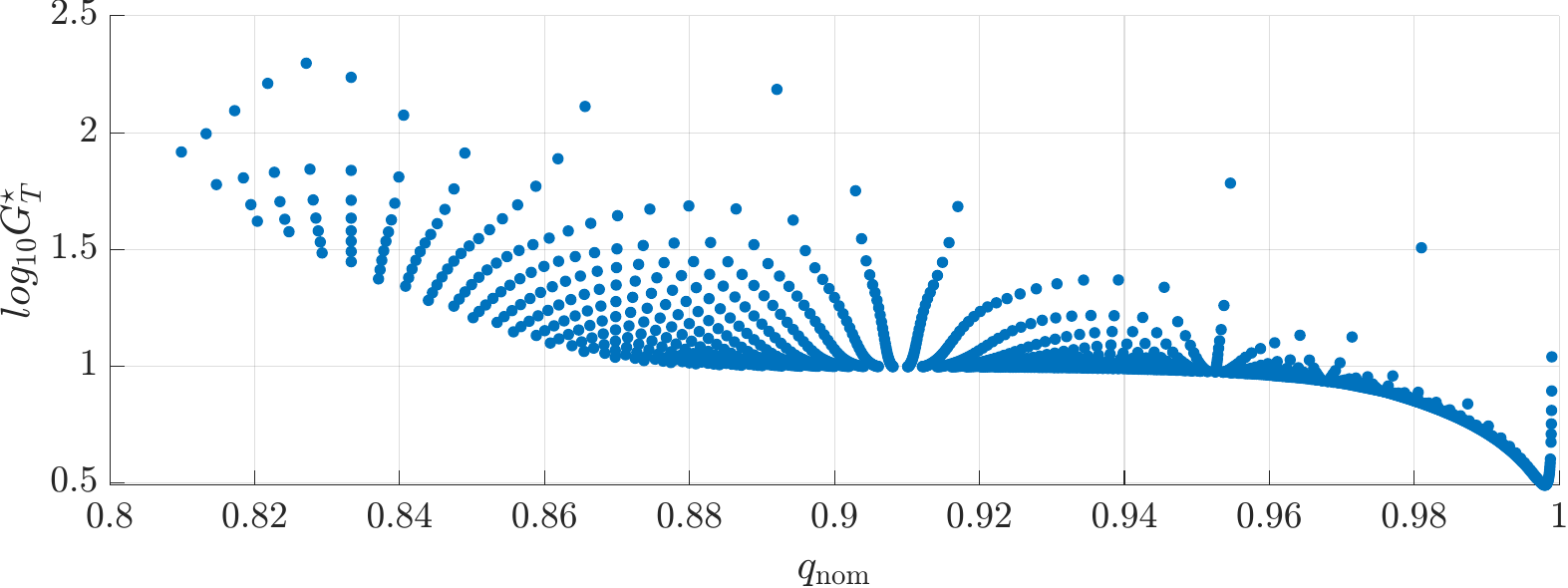}
	\fi
    \caption{Top: \gls{IQC}-certified finite-horizon gain $G_T^\star(\alpha,\beta)$ for \gls{HB} from the data-injection channel to the terminal current-iterate sensitivity. Bottom: certified sensitivity gain for \gls{HB} compared with a nominal quadratic convergence proxy.}
    \label{fig:hb-sensitivity-phase-diagram}
\end{figure}

\begin{figure}[t]
    \centering
	\ifTwoColumn
\else
		\includegraphics[width=.6\columnwidth]{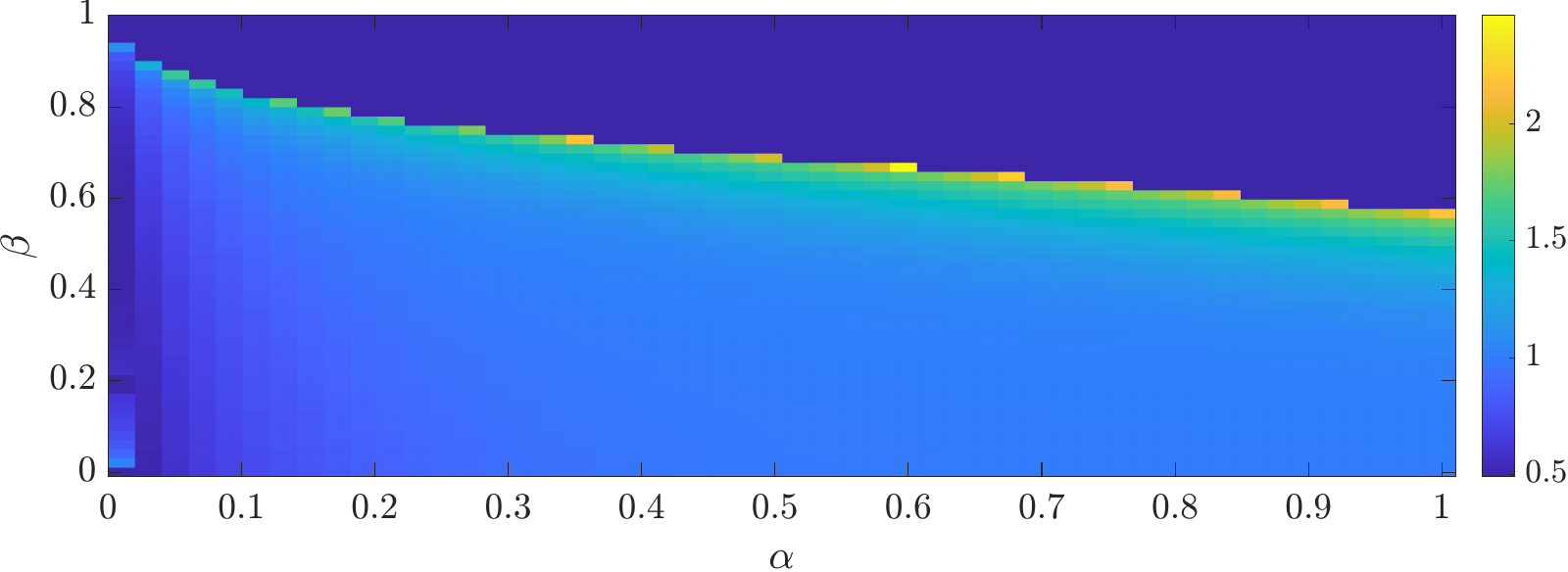}
    	\includegraphics[width=.6\columnwidth]{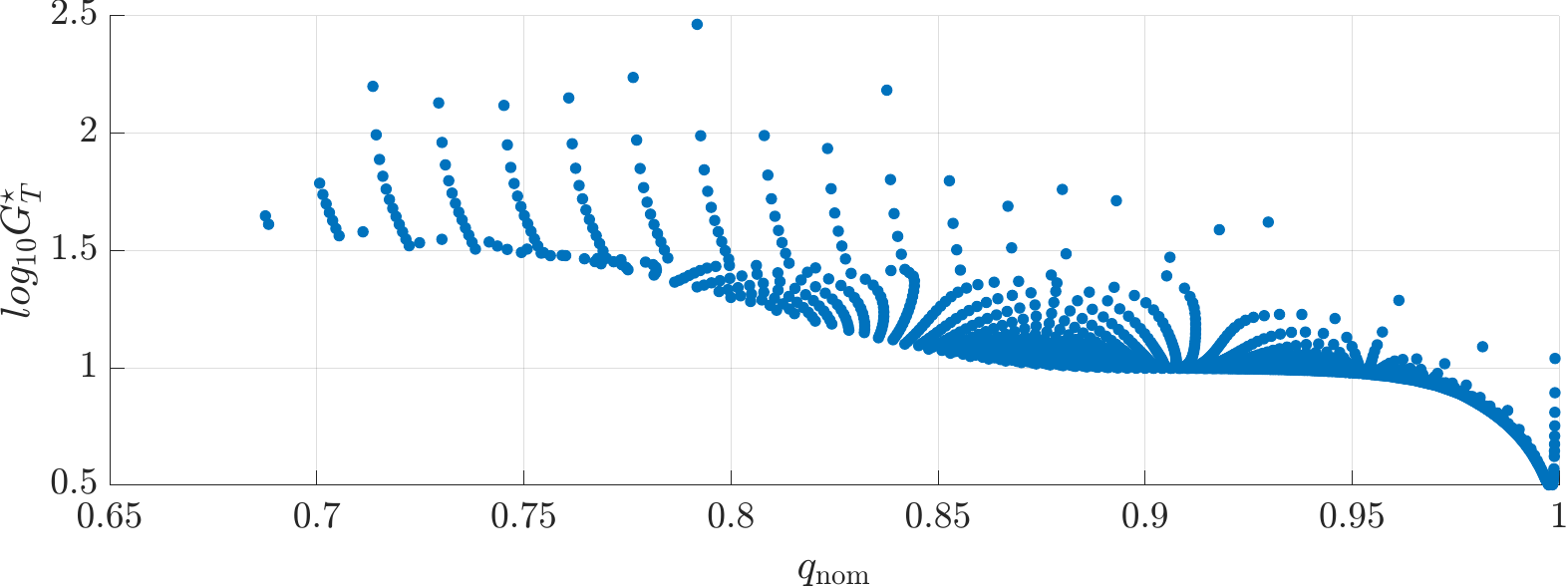}
	\fi
    \caption{Top: \gls{IQC}-certified finite-horizon gain $G_T^\star(\alpha,\beta)$ for \gls{NAG} from the data-injection channel to the terminal current-iterate sensitivity. Bottom: certified sensitivity gain for \gls{NAG} compared with a nominal quadratic convergence proxy.}
    \label{fig:nag-sensitivity-phase-diagram}
\end{figure}

Along the same line, Figs.~\ref{fig:hb-sensitivity-phase-diagram} and \ref{fig:nag-sensitivity-phase-diagram} report the certified finite-horizon gain $G_T^\star=G_T^\star(\alpha,\beta)\eqdef\sqrt{\gamma_2^\star(1 - \rho^{\star^{2T}})/(1-\rho^{\star^{2}})}$, $T=50$, which is further contrasted with $q_{\rm nom}$ in the scatter plot, for \gls{HB} and \gls{NAG}, respectively. Both heatmaps show that $G_T^\star(\alpha,\beta)$ significantly increases as the parameters approach the high-momentum boundary of the feasible region, with \gls{NAG} exhibiting a larger set of combinations of $(\alpha,\beta)$ guaranteeing the feasibility of \eqref{eq:gamma-star-rho}-\eqref{eq:G-star}. Even if the term $\zeta_N$ scales as $O(1/N)$, notice that the dynamical gain multiplying it can vary significantly across acceleration parameters, with larger values of $\alpha$ that substantially reduce the range of $\beta$ for which the certificate remains small. The scatter plots, instead, further confirm that nominal convergence and data sensitivity are distinct design objectives. 
Parameter choices yielding more favorable, i.e., smaller, $q_{\rm nom}$ may exhibit substantially larger certified sensitivity gain, whereas the smallest sensitivity gains occur in more conservative regimes where $q_{\rm nom}$ is close to one.

\section{Application to learned feedback control}\label{sec:learning_based_dyn}
When the data-driven map $\hat{\Phi}(\cdot,\mathcal{S}_N)$ is constructed to replicate a certain control action, its ability to generalize beyond the finite dataset $\mathcal{S}_N$ used for training is essential, since it is deployed online within the feedback interconnection in \eqref{eq:dynamics}.

Consequently, the relevant inputs to $\hat{\Phi}$ are not only the samples contained in $\mathcal{S}_N$, but the outputs generated by the closed-loop system itself. These outputs may differ from the empirical training distribution, and approximation errors of $\hat{\Phi}$ are recursively propagated through the dynamics. Therefore, a small empirical error on $\mathcal{S}_N$ is in general insufficient to guarantee satisfactory closed-loop behavior, since poor out-of-sample performance may lead to, e.g., degraded performance, constraint violations, or instability at worst.

This motivates the need for finite-sample generalization guarantees on the learned operator. For instance, by writing
\[
\hat{\Phi}(y,\mathcal{S}_N)=\Phi(y)+e_N(y),
\]
the learning error $e_N:\R^p\times\Xi^N\to\R^m$ enters the state equation as a feedback-dependent disturbance,
\[
x_{k+1}=Ax_k+B\Phi(y_k)+B e_N(y_k).
\]
Hence, stability and performance properties depend on controlling $e_N(y)$ over the set of outputs that can be visited by the closed-loop system, rather than only at the training samples $\mc S_N$. 
Generalization is thus the mechanism that connects finite offline data to reliable online closed-loop guarantees.

\subsection{\gls{KRR}-based residual feedback control}
\begin{figure}[t]
    \centering
	\ifTwoColumn
\else
		\includegraphics[width=.6\columnwidth]{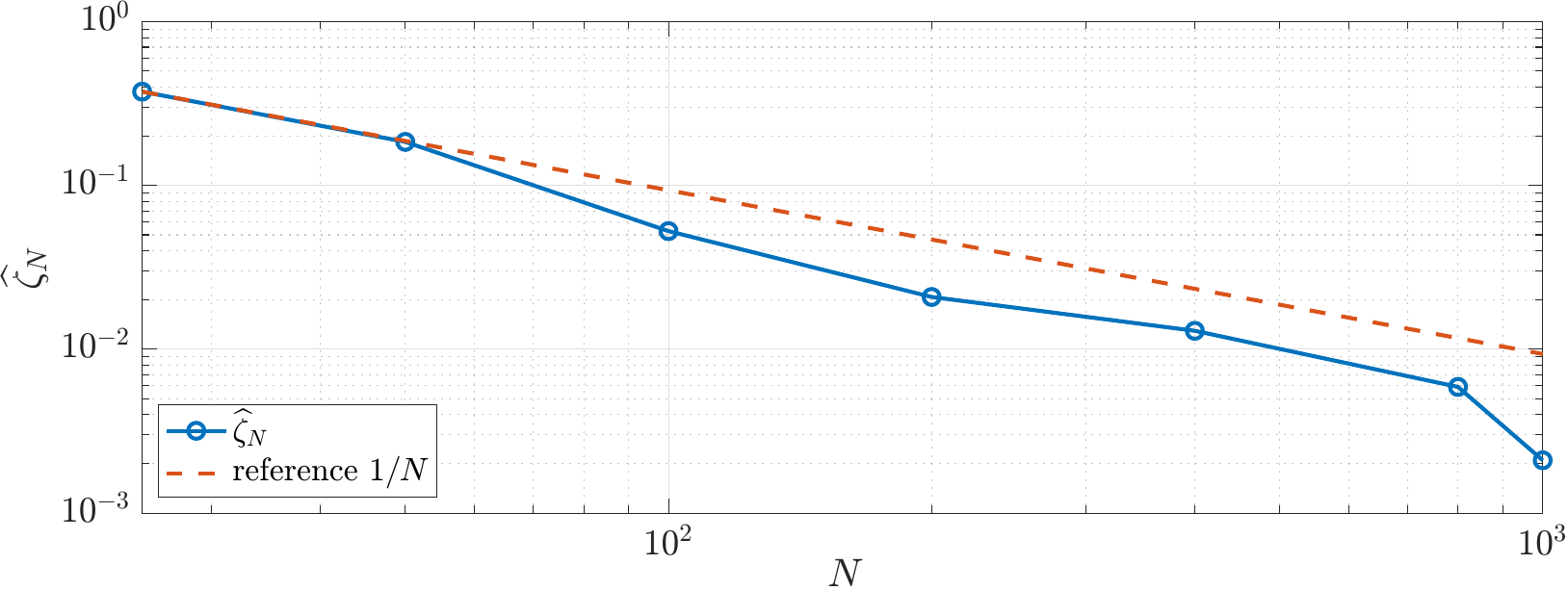}
			\includegraphics[width=.6\columnwidth]{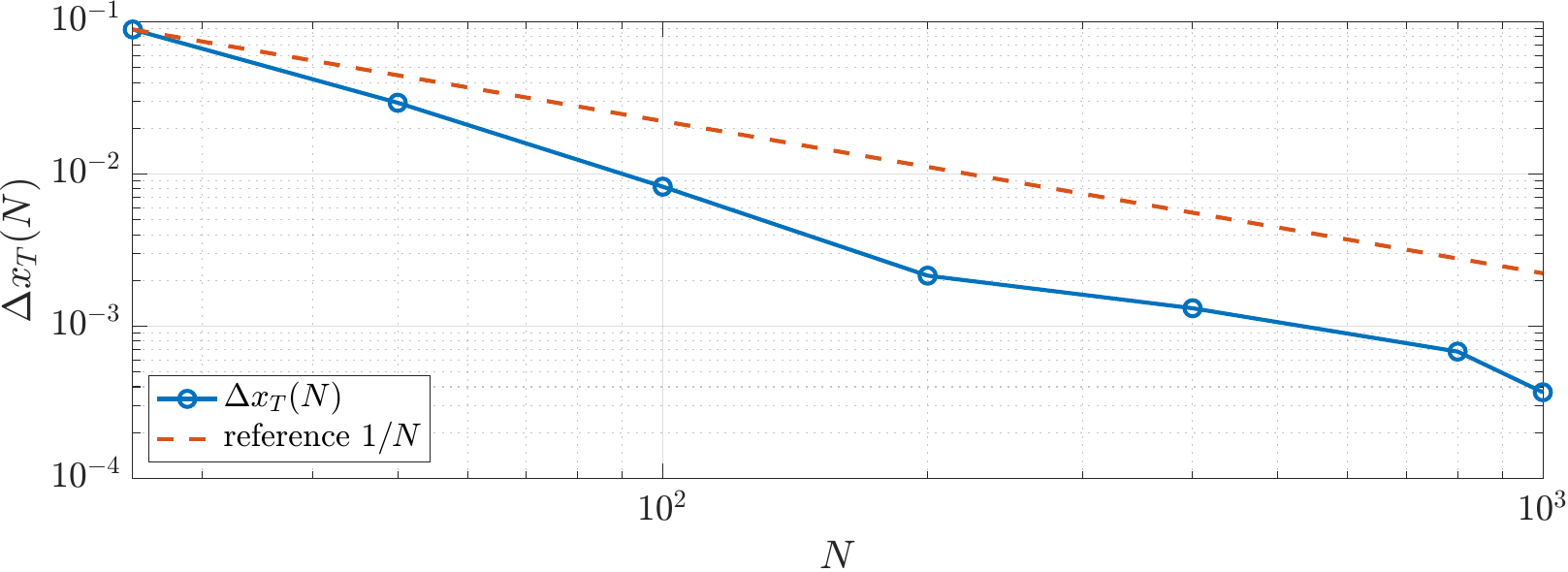}
	\fi
	\caption{  
    Top: Empirical one-sample sensitivity $\widehat\zeta_N$ for different sample sizes. Bottom: Terminal closed-loop sensitivity between trajectories generated by neighbouring learned controllers.}
    \label{fig:krr-control-sensitivity}
\end{figure}

We consider here a control-oriented example in which the data-dependent operator is learned offline from demonstrations. The plant is a discrete-time linear system $x_{k+1}=Ax_k+Bu_k$ equipped with a stabilizing nominal linear feedback $u=Kx$. In addition, we assume that an expert operator applies a nonlinear residual correction of the form $\mathrm{r}(x)=-a\tanh(v^\top x)$,
so that the measured input is actually $u^{(i)} = Kx^{(i)} + \mathrm{r}(x^{(i)})$.
The residual map is thus learned via \gls{KRR} from samples $\mc S_N=\{(x^{(i)},\mathrm{r}(x^{(i)})\}_{i=1}^N$,
yielding a data-dependent, state feedback residual controller $\hat\Phi(\cdot, \mc S_N)$. The learned closed loop dynamics hence reads as:
\[
    x_{k+1}
    =
    (A+BK)x_k+B\hat\Phi(x_k, \mc S_N).
\]
Replacing one sample then produces a neighbouring dataset $\mc S_N^j$, a different learned controller $\hat\Phi(\cdot, \mc S_N^j)$, and hence a different closed-loop trajectory. This is precisely the data sensitivity mechanism captured by our framework discussed in \S \ref{sec:sensitivity_approach}, where the perturbation
$
    \hat\Phi(x'_k, \mc S_N)
    -
    \hat\Phi(x'_k, \mc S_N^j)
$
acts as an exogenous disturbance injected into the closed-loop sensitivity dynamics.

The simulations are conducted by considering $n=2$, $m=1$ and linear controller $K=[-0.55 \ -0.85]$, with matrices
\[
A=\begin{bmatrix}
1.1 & 0.2\\
0 & 0.95	
\end{bmatrix},~B=\begin{bmatrix}
	0\\1
\end{bmatrix}.
\]
In addition, we set $a=0.6$ and $v=[1 \ -1]^\top$, and to fill the dataset $\mc S_N$ we uniformly sample the set $[-2 \ 2]^2\reqdef\mc G$ so that each $x^{(i)}\in\mc G$. For the \gls{KRR} technique, we use standard Gaussian radial basis functions with bandwidth $0.75$, and regularization term $0.001$.
For each \(N\in\{25, 50, 100, 200, 400, 800, 1000\}\), we estimate the empirical one-sample sensitivity over $30$ possible replacements of the learned residual controller as follows:
\[
    \widehat\zeta_N
    \eqdef
    \underset{j\in\{1,\ldots,N\}}{\textrm{max}}~\underset{x\in\mc G}{\textrm{max}}~
    |
    \hat\Phi(x, \mc S_N)
    -
    \hat\Phi(x, \mc S_N^j)
    |.
\]

Regardless of the specific sample replacement, both trajectories are initialized with $x_0=[1.5 \ -1]^\top$. For $T=40$, we thus estimate the corresponding terminal closed-loop sensitivity as:
\[
    \Delta x_T(N)
    \eqdef
    \underset{j\in\{1,\ldots,N\}}{\textrm{max}}~
    \left\|
    x_T-x'_T
    \right\|.
\]
The two quantities are reported in Fig.~\ref{fig:krr-control-sensitivity}, corroborating the statistical-dynamical decomposition underlying the proposed bound. From the top plot, the learned-controller sensitivity $\widehat\zeta_N$ decreases approximately with the reference rate \(1/N\), as predicted by the one-sample sensitivity of regularized kernel methods---see Corollary~\ref{cor:KRR}. The terminal closed-loop sensitivity exhibits the same sample-size decay, showing that the reduction in the statistical perturbation is inherited by the closed-loop trajectories. Even in a control setting with a nonlinear learned residual feedback, the observed terminal sensitivity follows the decay of the data injection channel.

\section{Conclusion}\label{sec:conclusion}
We developed a system-theoretic interpretation of generalization in learning-enabled dynamical systems by modeling sample replacement as an exogenous disturbance acting on a data sensitivity system. The operator increment is encoded through an \gls{IQC}, while dissipativity yields a finite-gain certificate from data injection to terminal loss discrepancy. The resulting uniform stability bound separates the learned operator’s one-sample sensitivity from the dynamical amplification induced by the interconnection. Computable via \gls{SDP}, the certificate provides a tractable tool for comparing learning dynamics. It recovers classical results for \gls{GD}, extends to \gls{HB} and \gls{NAG} when Euclidean contractivity fails, and shows that convergence and data sensitivity are distinct design objectives. A \gls{KRR} residual-control example further illustrates its applicability to operators learned offline and deployed in feedback.

Besides considering richer dynamic \gls{IQC} multipliers, nonquadratic storage functions, or focusing on more general nonlinear system realizations, an interesting direction to investigate consists in the algorithm synthesis for generalization. Since the algorithm realization enters the \gls{MI} \eqref{eq:main-lmi-induced-gain} directly,  one can in principle co-design parameters according to both convergence and data sensitivity objectives, as well as its generalization certificate, by means of a joint, possibly convex, program.

\bibliographystyle{IEEEtran}
\bibliography{algstab_dissipativity}

\appendix
\counterwithin{theorem}{section}

\section{Appendix}\label{app:1overN}
\subsection{One-sample sensitivity bounds for data-dependent operators}

\begin{table}[t]
\centering
\caption{Main examples of mappings satisfying
$\zeta_N\propto1/N$.}
\label{tab:bounds}
\begin{tabular}{p{2.8cm} p{2.5cm} p{2.4cm}}
\toprule
Mapping class & Assumption & Bound\\
\midrule
\gls{SAA} mapping & $\|\phi(y,\xi)\|\le B$ & $2B/N$\\
Empirical moments & bounded moments & $2L(\sum_k B_k^2)^{1/2}/N$\\
Regularized \gls{ERM} & $\mu$-strong convexity & $2B/\mu N$\\
KRR in RKHS
&
$\mathrm k(y,y)\le B^2$, $\norm{u}\le U$
&
$\frac{2 B^2U}{\mu N}
\left(1+\frac{B}{\sqrt{\mu}}\right)$\\
Contractive & contraction $q<1$ & $2LB/(1-q)N$\\
\bottomrule
\end{tabular}
\end{table}

Given a collection of samples $\mathcal S_N\in\Xi^N$, associated one-sample replacement dataset $\mathcal S_N^j$, for any $j\in\{1,\ldots,N\}$, and mapping $f:\R^d\to\R^n$, define the empirical averaging operator:
\[
\wexpected{\prob{P}}{f(\xi)}
\eqdef
\frac1N\sum_{i=1}^N f(\xi^{(i)}),
\]
and the corresponding empirical average after replacing the \(j\)-th sample by \(\xi'\):
\[
\wexpectedj{\prob{P}}{f(\xi)}
\eqdef
\frac1N
\left(
\sum_{i=1, i\neq j}^N f(\xi^{(i)})+f(\xi')
\right).
\]

The main classes and realated bounds are summarized in Tab.~\ref{tab:bounds}. We then start by considering \gls{SAA} mappings:
\begin{proposition}
Let us consider $\hat\Phi(y,\mathcal S_N)
=
g(
\wexpected{\prob{P}}{f_1(\xi)},
\ldots,
\wexpected{\prob{P}}{f_s(\xi)}
)$,
where each $f_k:\R^d\to\R^n$, and $g:\R^{sn}\to\R^m$ is $L$-Lipschitz. Assume that there exist constants \(B_1,\ldots,B_s>0\) such that $\|f_k(\xi)\|\le B_k$ for all $k=1,\ldots,s$ and $\xi\in\Xi$. 
Then, $\zeta_N\le2L
\sqrt{
\sum_{k=1}^s B_k^2}/N$.
\hfill$\square$
\end{proposition}

\begin{proof}
For every $k=1,\ldots,s$, we have
$
\wexpected{\prob{P}}{f_k(\xi)}-\wexpectedj{\prob{P}}{f_k(\xi)}
=
\frac1N
\left(
f_k(\xi^{(j)})-f_k(\xi')
\right),
$
which yields:
\[
	\|\wexpected{\prob{P}}{f_k(\xi)}-\wexpectedj{\prob{P}}{f_k(\xi)}\|
	\le
	\frac1N
	\left(
	\|f_k(\xi^{(j)})\|+\|f_k(\xi')\|
	\right)
	\le
	\frac{2B_k}{N}.
\]
Hence,
$
\sqrt{
\sum_{k=1}^s
\|\wexpected{\prob{P}}{f_k(\xi)}-\wexpectedj{\prob{P}}{f_k(\xi)}\|^2
}
\le
\frac{2}{N}
\sqrt{
\sum_{k=1}^s B_k^2}.
$
Using the $L$-Lipschitz continuity of $g$, one directly obtains:
\ifTwoColumn
\else
	\[
	\begin{aligned}
	\|\Phi(y,\mathcal S_N)-\Phi(y,\mathcal S_N^j)\|&=
	\left\|
	g\left(\wexpected{\prob{P}}{f_1(\xi)},\ldots,\wexpected{\prob{P}}{f_s(\xi)}\right)-
	g\left(\wexpectedj{\prob{P}}{f_k(\xi)},\ldots,\wexpected{\prob{P}}{f_s(\xi)}\right)
	\right\| \\
	&\le
	L
	\left(
	\sum_{k=1}^s
	\|\wexpected{\prob{P}}{f_k(\xi)}-\wexpectedj{\prob{P}}{f_k(\xi)}\|^2
	\right)^{1/2}\le
	\frac{2L}{N}
	\left(
	\sum_{k=1}^s B_k^2
	\right)^{1/2},
	\end{aligned}
	\]
\fi
concluding the proof.
\end{proof}

\begin{corollary}[\gls{SAA} mappings]\label{cor:SAA_mappings}
If $f(\xi)=\phi(y,\xi)$, with $\phi:\R^p\times\Xi\to\R^n$ and $\|\phi(y,\xi)\|\le B$, for all $y\in\R^p$, $\xi\in\Xi$, then $\zeta_N\le 2B/N$.
\hfill$\square$
\end{corollary}

\begin{corollary}[Lipschitz transformations of \gls{SAA}]
If $f(\xi)=\phi(y,\xi)$, with $\phi:\R^p\times\Xi\to\R^n$ and $\|\phi(x,\xi)\|\le B$, for all $y\in\R^p$, $\xi\in\Xi$, and $L$-Lipschitz $\wexpected{\prob{P}}{f(\xi)} \mapsto g(y,\wexpected{\prob{P}}{f(\xi)})$, then $\zeta_N\le2LB/N$.
\hfill$\square$
\end{corollary}

\begin{corollary}[Finite collections of empirical moments]
If $f_k(\xi)=\phi_k(y,\xi)$, with $\phi_k:\R^p\times\Xi\to\R^n$ and $\|\phi_k(y,\xi)\|\le B_k$, for all $y\in\R^p$, $\xi\in\Xi$, $k=1,\ldots, s$, and $L$-Lipschitz $(\wexpected{\prob{P}}{f_1(\xi)}, \ldots, \wexpected{\prob{P}}{f_s(\xi)}) \mapsto g(y,\wexpected{\prob{P}}{f_1(\xi)}, \ldots, \wexpected{\prob{P}}{f_s(\xi)})$, then $\zeta_N\le2L \sqrt{\sum_{k=1}^s B_k^2}/N$.
\hfill$\square$
\end{corollary}

\begin{corollary}[Kernel mean embeddings]
Let $\mathcal H$ be a Hilbert space. If $f(\xi)=k_y(\xi)\in\mathcal H$, with $\|k_y(\xi)\|_{\mathcal H}\le B$ for all $y\in\R^p$, $\xi\in\Xi$, then $\zeta_N\le 2B/N$.
\hfill$\square$
\end{corollary}

We now discuss the family of empirical minimizers:
\begin{proposition}\label{prop:EM}
Let $\theta\mapsto g(y,\theta)+\wexpected{\prob{P}}{\phi(y,\theta,\xi)}$ be $\mu$-strongly convex for all $y\in\R^p$, $\xi\in\Xi$, and replacement $\xi'\in\Xi$, with $g:\R^p\times\R^q\to\R$, $\phi:\R^p\times\R^q\times\Xi\to\R$, and $\hat\Phi(y,\mathcal S_N)\eqdef\mathrm{argmin}_{\theta\in\R^q}
\left\{g(y,\theta)+\wexpected{\prob{P}}{\phi(y,\theta,\xi)}
\right\}$. In addition, assume each sample loss being differentiable with uniformly bounded gradient, $\|\nabla \phi(y,\theta,\xi)\|\le B$, for all $y\in\R^p$, $\theta\in\R^q$, $\xi\in\Xi$. Then, $\zeta_N\le 2B/\mu N$.
\hfill$\square$
\end{proposition}

\begin{proof}
Let $\theta_N$ be the unique minimizers for $\theta\mapsto g(y,\theta)+\wexpected{\prob{P}}{\phi(y,\theta,\xi)}$, which satisfies:
$
	\nabla g(y,\theta_N)+\nabla\wexpected{\prob{P}}{\phi(y,\theta_N,\xi)}
	=\nabla g(y,\theta_N)+\frac1N\sum_{i=1}^N \nabla\phi(y,\theta_N,\xi^{(i)})=0.
$
Similarly, $\theta^j_N$ is the unique minimizer of $\theta\mapsto g(y,\theta)+\wexpectedj{\prob{P}}{\phi(y,\theta,\xi)}$ so that:
$
	\nabla g(y,\theta^j_N)+\nabla\wexpectedj{\prob{P}}{\phi(y,\theta^j_N,\xi)}
	=\nabla g(y,\theta^j_N)+\frac1N\left(\sum_{\stackrel{i=1}{i\neq j}}^N \nabla\phi(y,\theta^j_N,\xi^{(i)}) + \nabla\phi(y,\theta^j_N,\xi')\right)=0.
$
In view of the $\mu$-strong monotonicity of $\theta\mapsto \nabla g(y,\theta)+\nabla\wexpectedj{\prob{P}}{\phi(y,\theta,\xi)}$, consequence of the $\mu$-strong convexity of the underlying function, and the application of the Cauchy-Schwarz inequality we readily obtain:
\ifTwoColumn
\else
	\[
	\begin{aligned}
		\mu\|\theta_N-\theta_N^j\|&=\mu\|\hat\Phi(y,\mc S_N)-\hat\Phi(y,\mc S^j_N)\|\\
		&\le
		\left\|
		\nabla g(y,\theta_N)+\nabla\wexpectedj{\prob{P}}{\phi(y,\theta_N,\xi)}-\nabla g(y,\theta^j_N)-\nabla\wexpectedj{\prob{P}}{\phi(y,\theta^j_N,\xi)}
		\right\|\\
		&=\|\nabla g(y,\theta_N)+\nabla\wexpectedj{\prob{P}}{\phi(y,\theta_N,\xi)}\|.
	\end{aligned}
	\]
\fi
By exploiting the fact that $\nabla g(y,\theta_N)+\nabla\wexpected{\prob{P}}{\phi(y,\theta_N,\xi)}=0$:
\ifTwoColumn
\else
	\[
	\begin{aligned}
		\|\nabla g(y,\theta_N)+\nabla\wexpectedj{\prob{P}}{\phi(y,\theta_N,\xi)}\|&=\left\|\nabla g(y,\theta_N)+\nabla\wexpectedj{\prob{P}}{\phi(y,\theta_N,\xi)}\right.\\
		&\hspace{4cm}\left. -\nabla g(y,\theta_N) - \nabla\wexpected{\prob{P}}{\phi(y,\theta_N,\xi)}\right\|\\
		&\le
		\frac1N
		\left(
		\|\nabla\phi(x,\theta_N,\xi')\|
		+
		\|\nabla\phi(x,\theta_N,\xi^{(j)})\|
		\right)
		\le
		\frac{2B}{N},
	\end{aligned}
	\]
\fi
since the two gradients differ only from the replaced sample, directly yielding the desired bound on $\zeta_N$.
\end{proof}

\begin{corollary}[Regularized empirical minimizers]\label{cor:empirical_minimizer}
	If $g(y,\theta)=\mu\norm{\theta}^2/2$ and $\theta \mapsto \wexpected{\prob{P}}{\phi(y,\theta,\xi)}$ is convex for all $y\in\R^p$, $\xi\in\Xi$, and $\xi'\in\Xi$, then $\zeta_N\le 2B/\mu N$.
	\hfill$\square$
\end{corollary}

\begin{corollary}[KRR in RKHS]\label{cor:KRR}
	Let $\mathcal H$ be the \gls{RKHS} induced by a
	kernel $\mathrm k:\mathbb R^p\times\mathbb R^p\to\mathbb R_{\ge0}$,
	and let $\mathcal H^m$ denote the corresponding product space for
	vector-valued maps $f:\mathbb R^p\to\mathbb R^m$. In addition, 
	let
	$\mathcal Y\subseteq\mathbb R^p$ be such that $\textnormal{\textrm{sup}}_{y\in\mathcal Y}~\mathrm k(y,y) \le B^2$.
	With sample
	$\xi=\col(y,u)\in\mathcal Y\times\mathbb R^m$ satisfying $\|u\|\le U$,
	for $\mu>0$ we define the \gls{KRR} estimator as
	\[
	\hat f_N
	\eqdef
	\underset{f\in\mathcal H^m}{\textnormal{\textrm{argmin}}}~
	\left\{
	\frac{\mu}{2}\|f\|_{\mathcal H^m}^2
	+
	\frac{1}{2N}\sum_{i=1}^N
	\|f(y^{(i)})-u^{(i)}\|^2
	\right\}.
	\]
	With $\hat\Phi(y,\mathcal S_N)\eqdef\hat f_N(y)$, for every neighboring dataset $\mathcal S_N^j$,
	\[
	\underset{y\in\mathcal Y}{\textnormal{\textrm{sup}}}~
	\|\hat\Phi(y,\mathcal S_N)-\hat\Phi(y,\mathcal S_N^j)\|
	\le
	\frac{2B^2 U}{\mu N}
	\left(
	1+\frac{B\sqrt{\mu}}{\mu}
	\right).
	\]
	\hfill$\square$
\end{corollary}

\begin{proof}
	Let $\hat f_N$ and $\hat f_N^j$ denote the KRR minimizers associated
	with $\mathcal S_N$ and $\mathcal S_N^j$, respectively. Then, define the cost function as
	$
	J_N(f)
	\eqdef
	\mu\|f\|_{\mathcal H^m}^2/2
	+
	\sum_{i=1}^N
	\|f(y^{(i)})-u^{(i)}\|^2/2N,
	$
	along with the analogous counterpart $J_N^j$ obtained after replacing
	$\xi^{(j)}=\col(y^{(j)},u^{(j)})$ with
	$\xi'=\col(y',u')$. Since $J_N$ is $\mu$-strongly convex on
	$\mathcal H^m$, its minimizer is unique. Moreover,
	$
	J_N(\hat f_N)\le J_N(0)
	=
	\frac{1}{2N}\sum_{i=1}^N \|u^{(i)}\|^2
	\le \frac{U^2}{2},
	$
	which directly implies
	$
	\|\hat f_N\|_{\mathcal H^m}\le U/\sqrt{\mu}.
	$
	For a generic sample $\xi=\col(y,u)$, let
	$
	\phi(f,\xi)\eqdef\|f(y)-u\|^2/2.
	$
	By the reproducing property, the gradient of $\phi$ with respect to
	$f\in\mathcal H^m$ satisfies
	\[
	\|\nabla \phi(\hat f_N,\xi)\|_{\mathcal H^m}
	\le
	\sqrt{\mathrm k(y,y)}
	\,
	\|\hat f_N(y)-u\|.
	\]
	Since $y\in\mathcal Y$, $\mathrm k(y,y)\le B^2$.
	Again by the reproducing property,
	\[
	\|\hat f_N(y)\|
	\le
	B\|\hat f_N\|_{\mathcal H^m}
	\le
	\frac{BU}{\sqrt{\mu}},
	\]
	which yields
	$
	\|\nabla \phi(\hat f_N,\xi)\|_{\mathcal H^m}
	\le
	BU
	\left(
	1+\frac{B}{\sqrt{\mu}}
	\right).
	$

	Since $\hat f_N^j$ minimizes $J_N^j$ and $J_N^j$ is $\mu$-strongly
	convex, the same argument used in Proposition~\ref{prop:EM} gives
	\[
	\mu\|\hat f_N-\hat f_N^j\|_{\mathcal H^m}
	\le
	\|\nabla J_N^j(\hat f_N)-\nabla J_N(\hat f_N)\|_{\mathcal H^m}.
	\]
	Because $J_N$ and $J_N^j$ differ only in the replaced sample,
	\ifTwoColumn
\else
		\[
			\|\nabla J_N^j(\hat f_N)-\nabla J_N(\hat f_N)\|_{\mathcal H^m}\le
			\frac{1}{N}
			\left(
			\|\nabla\phi(\hat f_N,\xi')\|_{\mathcal H^m}
			+
			\|\nabla\phi(\hat f_N,\xi^{(j)})\|_{\mathcal H^m}
			\right)\le
			\frac{2BU}{N}
		\left(
		1+\frac{B}{\sqrt{\mu}}
		\right),
		\]
	\fi
	and thus
	$
	\|\hat f_N-\hat f_N^j\|_{\mathcal H^m}
	\le
	\frac{2BU}{\mu N}\left(
	1+\frac{B}{\sqrt{\mu}}
	\right).
	$
	Finally, using the reproducing property once more, for every
	$y\in\mathcal Y$,
	\[
	\|\hat\Phi(y,\mathcal S_N)-\hat\Phi(y,\mathcal S_N^j)\|
	=
	\|\hat f_N(y)-\hat f_N^j(y)\|
	\le
	B\|\hat f_N-\hat f_N^j\|_{\mathcal H^m},
	\]
	which allows one to establish the desired bound.
\end{proof}

We conclude by discussing the case in which $\hat\Phi$ consists of a contractive mapping driven by empirical averages.
\begin{proposition}
	Suppose that $\hat\Phi(y,\mathcal S_N)=z_N$ where \(z_N\) is the unique solution to $z=\mc T\left(y,z,\wexpected{\prob{P}}{\phi(y,\xi)}\right)$. Assume that $\mc T$ is contractive in $z$, i.e., there exists $q\in[0,1)$ such that
	\[
	\|\mc T(y,z,u)-\mc T(y,z',u)\|
	\le
	q \|z-z'\|,
	\]
	for all $y\in\R^p$, $u\in\R^m$, and $u\mapsto \mc T(y,z,u)$ being $L$-Lipschitz. If, in addition, $\|\phi(y,\xi)\|\le B$, for all $y\in\R^p$, $\xi\in\Xi$, then $\zeta_N\le 2LB/(1-q)N$.
	\hfill$\square$
\end{proposition}

\begin{proof}
In view of Corollary~\ref{cor:SAA_mappings}, $\|\wexpected{\prob{P}}{\phi(y,\xi)}-\wexpectedj{\prob{P}}{\phi(y,\xi)}\|\le 2B/N$.
Let $z_N$ and $z_N^j$ be the fixed points corresponding to $\mc T(y,z,\wexpected{\prob{P}}{\phi(y,\xi)})$ and $\mc T(y,z,\wexpectedj{\prob{P}}{\phi(y,\xi)})$, respectively. Thus, we immediately obtain inequalities:
\ifTwoColumn
\else
	\[
	\begin{aligned}
	\|z_N-z_N^j\|
	&=
	\|\mc T(y,z_N,\wexpected{\prob{P}}{\phi(y,\xi)})-\mc T(y,z^j_N,\wexpectedj{\prob{P}}{\phi(y,\xi)})\| \\
	&\le
	\|\mc T(y,z_N,\wexpected{\prob{P}}{\phi(y,\xi)})-\mc T(y,z^j_N,\wexpected{\prob{P}}{\phi(y,\xi)})\| \\
	&\hspace{6cm}+
	\|\mc T(y,z^j_N,\wexpected{\prob{P}}{\phi(y,\xi)})-\mc T(y,z^j_N,\wexpectedj{\prob{P}}{\phi(y,\xi)})\| \\
	&\le
	q\|z_N-z_N^j\|
	+
	L\|\wexpected{\prob{P}}{\phi(y,\xi)}-\wexpectedj{\prob{P}}{\phi(y,\xi)}\|.
	\end{aligned}
	\]
\fi
Rearranging the terms, $(1-q)\|z_N-z_N^j\|\le L\|\wexpected{\prob{P}}{\phi(y,\xi)}-\wexpectedj{\prob{P}}{\phi(y,\xi)}\|\le 2LB/N$, which together with $\|z_N-z_N^j\|=\|\hat\Phi(y,\mathcal S_N)-\hat\Phi(y,\mathcal S^j_N)\|$ concludes the proof.
\end{proof}


\end{document}